\RequirePackage[2018-12-01]{latexrelease}
\documentclass{article}
\usepackage{geometry}
\usepackage{amsmath,amsthm,amssymb,amsfonts,graphicx,dsfont}
\usepackage{booktabs}
\usepackage{epsfig}
\usepackage{url}
\usepackage{bm}
\usepackage{rotating}
\usepackage{listings}
\usepackage{verbatim}
\usepackage{xcolor}
\usepackage{multirow}
\usepackage{adjustbox}
\usepackage{float} 
\usepackage{subfigure}
\usepackage{setspace} 

\newcommand{\eval}[2][\right]{\relax
	\ifx#1\right\relax \left.\fi#2#1\rvert}

\newcommand{\R}{{\mathbb R}}

\newcommand{\Dcal}{{\mathcal D}}

\newcommand{\Ical}{{\mathcal I}}

\newcommand{\Rcal}{{\mathcal R}}

\newcommand{\Wcal}{{\mathcal W}}

\newenvironment{keywords}{
  \noindent\textbf{Keywords:}
  \itshape
}{\par}

\newtheoremstyle{break}
  {\topsep}{\topsep}%
  {\itshape}{}%
  {\bfseries}{}%
  {\newline}{}%
\theoremstyle{break}

\theoremstyle{plain}
\newtheorem{theorem}{Theorem}

\newtheorem{prop}{Proposition}
\theoremstyle{definition}
\newtheorem{defn}{Definition}

\theoremstyle{remark}
\newtheorem{rem}{Remark}

\usepackage{authblk}

\title{Sparse Portfolio Selection via Topological Data Analysis-Based Clustering}
\author[1]{Anubha Goel}
\author[2]{Damir Filipovi\'c}
\author[1,3]{Puneet Pasricha\thanks{Corresponding author. Email: puneet.pasricha@epfl.ch}}

\affil[1]{\'Ecole Polytechnique F\'ed\'erale de Lausanne}
\affil[2]{\'Ecole Polytechnique F\'ed\'erale de Lausanne and Swiss Finance Institute}
\affil[3]{Department of Mathematics, Indian Institute of Technology Ropar, India}

\begin{document}
\maketitle

\onehalfspacing



\begin{abstract}
This paper uses topological data analysis (TDA) tools and introduces a data-driven clustering-based stock selection strategy tailored for sparse portfolio construction. Our asset selection strategy exploits the topological features of stock price movements to select a subset of topologically similar (different) assets for a sparse index tracking (Markowitz) portfolio. We introduce new distance measures, which serve as an input to the clustering algorithm, on the space of persistence diagrams and landscapes that consider the time component of a time series. We conduct an empirical analysis on the S\&P index from 2009 to 2022, including a study on the COVID-19 data to validate the robustness of our methodology. Our strategy to integrate TDA with the clustering algorithm significantly enhanced the performance of sparse portfolios across various performance measures in diverse market scenarios.
\end{abstract}


\section{Introduction}

Asset selection and portfolio optimization are critical tasks for financial institutions such as banks, fund management firms, and individual investors as they seek to manage risk and maximize returns in an increasingly competitive business environment and complex market conditions. Over the years, academics and practitioners have actively developed quantitative techniques and models to assess risk and construct efficient portfolios that perform well out-of-sample. In this pursuit, sparse portfolio selection has become a well-known and extensively researched problem in the fields of financial economics and optimization, where sparseness is a desirable property not only from the point of view of transaction costs but also because it allows investors to concentrate their capital on assets that are most likely to generate returns while minimizing exposure to unforeseen risks or market fluctuations. Moreover, it also offers a potential solution to the limitations of traditional mean-variance (\cite{dai2018some,gunjan2023brief}) or index-tracking portfolios, as it focuses on a smaller subset of assets to improve portfolio performance. Recently, data-driven approaches (\cite{benidis2017sparse,baes2021low,li2022sparse}) have been successfully developed for active and adaptive strategies on both index-tracking and mean-variance problems with sparse portfolios.\footnote{Despite its popularity, mean-variance analysis has received criticism for its susceptibility to parameter estimates and assumptions. It is observed that small changes in assumed asset returns, volatilities, or correlations lead to substantial effects on the optimization results, ultimately leading to sub-optimal performance when tested against real-world data. In essence, the classic Markowitz portfolio optimization represents an ill-posed inverse problem. Several techniques have been proposed to address this challenge, including regularization techniques imposing sparseness (\cite{ban2018machine,kremer2020sparse}) and adopting shrinkage estimators of expected returns and covariance matrix (\cite{lee2020sparse,pun2019linear,chen2023distributionally}).}$^{,}$\footnote{In index tracking, the goal is to construct a portfolio of securities that replicates the performance of a market index. The simplest method to achieve this goal is through full replication, which entails holding all the assets in the target market index in the same proportions. While full replication provides perfect tracking performance in a frictionless market, it is expensive, requiring frequent re-balancing of many stocks. It is sometimes not possible due to restrictions on trading illiquid securities or managerial constraints. An alternative approach is partial replication, which aims to mimic the index performance using a small subset of securities.} 

The problem of sparse portfolio selection can be reformulated as a two-step problem, with the first step of selecting a subset of securities based on their risk and return characteristics and the second step of obtaining the optimal weights. In this direction, clustering methods have emerged as one of the most innovative strategies in portfolio selection in recent years. This technique for classifying data objects into homogeneous groups based on their similarities proved to be effective in constructing sparse portfolios. For instance, \cite{panton1976comovement} studied the comovement of international equity markets using cluster analysis, although their goal was to study the dynamics of this structural comovement and not portfolio optimization. \cite{hui2005portfolio} used factor analysis to screen out certain stock markets before creating a Markowitz portfolio. More recently, \cite{leon2017clustering}
tested several clustering algorithms, including hierarchical clustering, using correlation coefficients as similarity measure to group stocks and create optimal portfolios. In this paper, we adopt clustering to select a smaller set of securities that provide diversification benefits to the portfolio.

Other methods in the literature for sparse portfolio selection includes regularization techniques (imposing a constraint on the $\ell_q$-norm of portfolio weights) (refer \cite{brodie2009sparse, demiguel2009generalized, wu2024sparse}) or including a cardinality constraint on the number of securities to be included in the portfolio to promote sparsity and stable allocations (\cite{anagnostopoulos2011mean,kalayci2020efficient,shi2022cardinality,xu2024efficient}). \cite{brodie2009sparse} was the first to integrate an $\ell_1$-norm (LASSO) penalty on the vector of portfolio weights in optimization models to regularize the optimization problem and encourage sparse portfolios. \cite{demiguel2009generalized} demonstrate that norm-constrained global minimum variance portfolio yields improved out-of-sample performance. \cite{giamouridis2010regular} and \cite{giuzio2018tracking} utilize $\ell_1$ penalty in the context of replicating hedge funds. \cite{ho2015weighted} considered the mean-variance optimization with an advanced regularization technique called weighted Elastic-Net. The $\ell_1$ penalty and its variants have also been studied in \cite{kremer2020sparse,corsaro2019adaptive} to obtain sparse portfolios. Recently, \cite{wu2024sparse} considered fraction regularization $\frac{\ell_1}{\ell_2}$ on the mean–variance model to promote sparse portfolio selection.

We highlight that the clustering based two-step approach has several advantages over the above-mentioned alternative methods. For instance, adding a cardinality constraint significantly increases the combinatorial complexity of the optimization problem, resulting in a mixed-integer quadratic programming problem that is generally NP-hard and computationally intractable. Heuristic methods, such as genetic algorithms, are typically required to find approximate solutions, which may not guarantee global optimality (\cite{ruiz2009hybrid,woodside2011heuristic,fastrich2014cardinality}). Similarly, the regularization technique involves determining the regularization coefficients through cross-validation, which becomes particularly challenging when each security has its regularization parameter; the approach  that has been shown to outperform the approach with a common regularization parameter (\cite{ho2015weighted}). Further, these methods tend to result in risk concentration as the resulting portfolio may not be diversified. This can lead to over-exposure to a single industry sector, thus making the portfolio riskier due to vulnerability to a downturn in that sector (\cite{kremer2020sparse,wu2024sparse}). Furthermore, due to local optima and discontinuities in the search space, both the cardinality constraint problem and the $\ell_q$-norm constraint problem are difficult to solve from an optimization standpoint. Finally, the regularization technique LASSO suffers from several known shortcomings, such as selecting at random from equally correlated assets (\cite{bondell2008simultaneous}), being biased for large coefficient values (\cite{fan2001variable,giuzio2018tracking}), and being ineffective in the no-short-sale area (\cite{demiguel2009generalized}). In contrast, clustering can result in a more diversified and sparse portfolio without increasing the computational burden. Moreover, clustering allows for the combination of judgment and quantitative optimization as it is an intuitive approach.

The key to any clustering technique is the distance function that enables it to identify similarities and dissimilarities among the inputs in the data set (time series of stock returns in our case) and allocate them to different clusters. Traditionally, correlation-based clustering has been employed to locate clusters in the portfolio of financial assets, for instance, \cite{mantegna1999hierarchical,leon2017clustering}. However, it is well known that Pearson’s correlation is not a suitable measure of dependence to be used outside the Gaussian (elliptical) distributions, and, in particular, it does not give an accurate indication and understanding of the actual dependence between risk exposures (\cite{embrechts2003using}). Furthermore, even if two assets are perfectly correlated, an investment in one can lead to positive and negative returns in the other. In this paper, we employ a recently introduced approach known as Topological Data Analysis (TDA) to define the similarity between two time series and propose an efficient data-driven approach to sparse portfolio selection.

TDA offers a novel tool to study the shape and qualitative properties of complex datasets which often have non-linear relationships and highly complex dependence structure. The core idea behind TDA is that data has a shape that conveys valuable information.\ TDA describes the ``shape" of the data by extracting topological features, characterized by $k$-dimensional holes, where $k=0,1,2$ represents connected components, holes, and voids, respectively. TDA, by focusing on the topological features of the data, captures information that traditional methods such as correlation may miss. Thus, studying the latent geometric behavior of assets could reveal clusters of assets with similar dynamics. What's more, a key advantage of TDA is its capability to extract topological features with minimum statistical assumptions, and the only requirement being the computation of distances between various points in the point cloud.

Persistence homology is one of the powerful tools in TDA that provides a comprehensive understanding of the structure of complex data by analyzing its shape across multiple resolutions while simultaneously tracking the changes in its topological features.  The information on all topological features with their birth and death is then summarized in a persistence diagram. Notably, persistence diagram, obtained using a point-cloud corresponding to a time-series, can capture various features of the data. For instance, higher fluctuations in time-series often result in more complex topological structures, while structural breaks or different regimes can alter the shape of the data. Further, the higher the concentration (scattering) of the point cloud, and
thus the more (less) stable the returns are over the sub-intervals from a topological point of view. By this way, the topological properties of the return data are learnt to capture the variability  of returns from the topological point of view.

A comparison among persistence diagrams corresponding to different assets could reveal interesting insights that correlation may miss. While correlation measures linear dependence time-series of two assets, TDA focuses on their qualitative topological properties. TDA based clusters may align with those obtained using correlation, but the former can reveal additional insights particularly in scenarios with non-linear dependencies or structural breaks, a typical behavior observed the financial assets. Thus, clustering persistence diagrams could help identify clusters of assets with similar risk-return properties, leading to more informed portfolios.


Recent developments in TDA have shown significant potential in extracting topological features of data, leading to an improved performance when combined with other models, such as traditional statistical models or machine learning models. The foundation of TDA was laid down by the pioneer contributions of \cite{edelsbrunner2000topological} and \cite{zomorodian2005computing}, while an overview paper by \cite{carlsson2009topology} turned out to be cornerstone. Thereafter, TDA has been widely applied in various fields, such as time-series clustering and classification (\cite{majumdar2020clustering,karan2021time}), computer science 
(\cite{pereira2015persistent}), biology (\cite{lum2013extracting}), semi-conductors manufacturing (\cite{ko2023novel}). TDA only started seeing application in finance very recently. For example, \cite{gidea2017topological} and \cite{gidea2018topological} used TDA to detect crash early warnings, \cite{goel2020topological} applied persistence landscapes to create portfolios. \cite{wu2022topological}, \cite{ferri2018topology} and \cite{moroni2021learning} integrated TDA and machine learning with persistence diagrams and studied their applications in various fields including finance. \cite{qiu2020refining} studied corporate failure by mapping Altman's z-score through TDA, while \cite{ofori2021nonparametric} used TDA to detect anomalies in time-series of graphs. Recently, \cite{tdapairtrading} proposed TDA-based distance to measure dependence between stochastic processes and used those distance measures in pair trading. Motivated by the success of TDA, this paper argues that topological features can help identify exciting time series patterns and spatial data clustering. In particular, insights can be gained by examining the essential qualitative characteristics of data sets that are often overlooked by traditional distance-based clustering techniques.  To our knowledge, this paper is the first to present a pioneering application of TDA to address the sparse portfolio selection using clustering.

Our paper makes several methodological and empirical contributions that address the aforementioned issues. First, despite the considerable amount of research that has been conducted on applying TDA to study financial time-series, there still needs to be practical implementation and utilization of this method in the field of quantitative finance.\footnote{Recently, \cite{goel2023sparse} proposed TDA based sparse index tracking model where the authors adopt the regularization approach to achieve sparseness.} We leverage tools from TDA, namely persistence diagrams and landscapes, and design strategies for sparse portfolio selection. By doing so, we establish a link between the growing research on TDA and the most frequently studied portfolio selection problem in the financial economics literature. Thus, we contribute to the emerging literature on applications of TDA to the problem of selecting sparse portfolios.

Second, our research draws attention to a critical limitation of distance metrics defined on persistence diagrams when applied to financial time-series data. Specifically, we show that the Wasserstein distance between the PDs corresponding to the return series of two hypothetical assets is zero, even when the two return series are not identical. We identify the reason for this behavior as the critical time component in financial data. To address this issue, we introduce two novel distance measures. The first, called the differenced Wasserstein distance (DWD), defines the distance between the persistence diagrams of two time series, ${x_t}$ and ${y_t}$, as the Wasserstein distance between the persistence diagram of the differenced time series ${x_t-y_t}$ and that of a constant series. The second called the average Wasserstein distance (AWD), partitions the two time series into overlapping sub-series and calculates the average Wasserstein distance between the corresponding sub-series. We also introduce similar measures, including the average landscape distance (ALD) and the differenced landscape distance (DLD) on the space of persistence landscapes. Moreover, we also demonstrate the stability of the persistence diagrams concerning the newly introduced distance measures, thus contributing to the literature on TDA.

Third, we utilize the introduced distance measures to propose two strategies for sparse portfolio selection. Our first strategy focuses on constructing a sparse index tracking (IT) portfolio, while our second strategy aims to select sparse mean-variance (MV) and global minimum variance (GMV) portfolios. To accomplish this, we create a distance matrix corresponding to the index and its $n$ constituents and define the corresponding similarity matrix. We then use this matrix as input for affinity propagation clustering (APC). We select the assets that share the cluster with the index for the sparse IT portfolio, which are topologically similar to the index by construction. We solve an index tracking problem on the selected assets. We select representative assets from each cluster for the sparse MV and GMV portfolios and solve the portfolio optimization problem on the filtered assets, thus obtaining a sparse portfolio. Our TDA-based similarity matrices and the APC offer an entirely data-driven approach, unlike the methodologies that employ $k$-mean clustering for sparse portfolio selection, which requires investors to pre-specify the number of assets in the portfolio.

We conduct an extensive empirical analysis, investigating the performance of the sparse index tracking, MV, and GMV portfolios obtained using the daily returns of the S\&P 500 index and its constituents spanning more than a decade from Dec 2009 to July 2019. We further assess the performance of our strategies during the recent COVID-19 Pandemic period from August 2019 to August 2020 to check the robustness. We further utilize the period, September 2020 to August 2022, for comparison
purposes, in particular, to other clustering approaches such as $K$-medoids and hierarchical clustering. We also compare the performance of our sparse index tracking model on the benchmark data set for index tracking problems available publicly on the Beasley OR library. Our analysis shows that our TDA-based sparse portfolios outperform the traditional correlation-based sparse portfolios as well as other leading benchmark models in the literature.

The rest of this paper is structured as follows. Section \ref{background} briefly reviews TDA. Section \ref{newdistances} introduces the novel distance measures and demonstrates their stability. Section \ref{our_methodology} describes the methodology, starting with the problem of sparse portfolio selection, followed by a brief review of affinity propagation clustering, and finally, our portfolio selection strategies. Section \ref{empirical_results} presents the data and the empirical findings. Section \ref{conclusion} concludes with potential future research directions.

\section{Background of TDA}\label{background}

Inspecting the 2-dimensional point clouds in Figure \ref{example_pointcloud}, it is evident that each point cloud has a distinct shape.{\footnote{A point cloud is a finite set of points in some Euclidean space $\mathbb{R}^d$ (or a more general metric space).}} Specifically, the points in the first point cloud are scattered around one loop, while those in the second point cloud are dispersed around two loops. Given our inability to visualize beyond two or three dimensions, the natural question is how to formalize these qualitative observations, particularly the shapes of point clouds in higher dimensions. Fortunately, TDA provides a rigorous tool to systematically define and study the ``shape" of data in high dimensional spaces.

\begin{figure*}[t!]
    \centering
    \includegraphics[height=4.5cm,width=9cm]{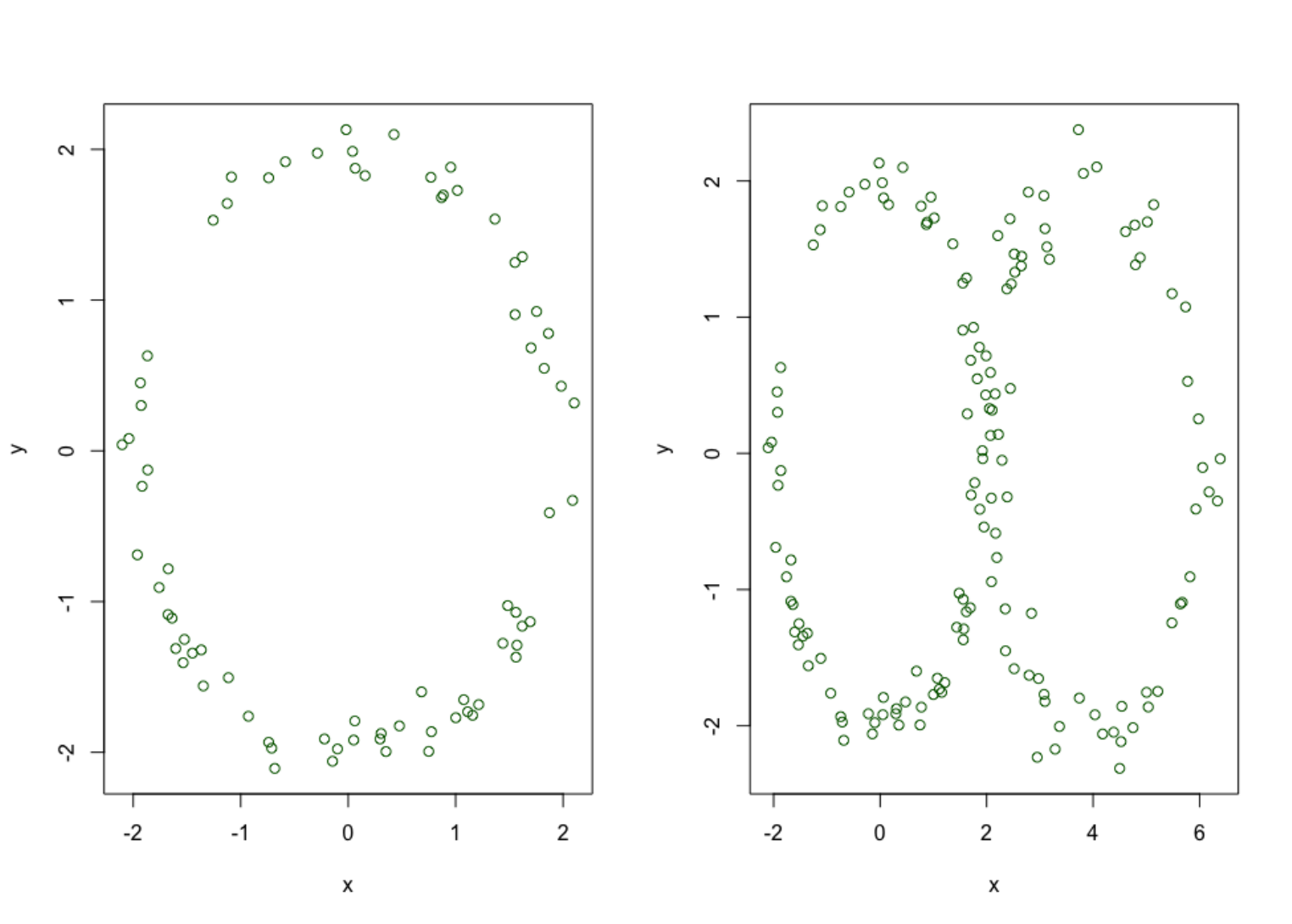}
    \caption{Examples of two-dimensional noisy point clouds.}
    \label{example_pointcloud}
\end{figure*}

The first step involves building a simplicial complex given the point cloud $X=\{x_i\}^N_{i=0}$ in $\mathbb{R}^q$. A simplicial complex is a space comprising points, edges, triangles, tetrahedra, and higher-dimensional polytopes. In the abstract sense, a simplicial complex is a finite set $S$ of non-empty subsets of a set $S_0$ such that $\{v\} \in S$ for every $v \in S_0$, the intersection of any two simplices in $S$ is either empty or shares faces, and $\tau \subset \sigma$ and $\sigma \in S$ ensure that $\tau \in S$.{\footnote{A $k$-simplex is a convex hull $[a_0, a_1, \ldots , a_k]=\{\sum_{i=0}^k\eta_ia_i: \sum_{i=0}^k\eta_i=1\}$ of $k+1$ geometrically independent points $\{a_0, a_1, \ldots , a_k\}$. Furthermore, the faces of a $k$-simplex $[a_0,\ldots,a_k]$ are the $(k - 1)$-simplices spanned by subsets of $\{ a_0,\ldots,a_k\}$. Note that 0-dimensional simplices represent data points (vertices), 1-dimensional simplices are connected pairs of vertices (edges), and a 2-dimensional simplex is a filled triangle determined by its three vertices.}} There are several ways to construct a simplicial complex for a given point cloud, including the Cech complex and the Vietoris-Rips complex. We adopt the Vietoris-Rips complex in this paper. The underlying idea of the Rips complex is to connect the pairs that are sufficiently close to points (vertices) by edges. Mathematically, 
\begin{defn}
Given a resolution parameter $\epsilon >0$, the Vietoris--Rips complex  of $X$ is defined to be  the simplicial complex $\Rcal_{\epsilon}(X)$ satisfying $[x_{i1},\ldots,x_{il}]\in \Rcal_{\epsilon}(X)$ if and only if $diam(x_{i1},\ldots,x_{il})< \epsilon$. 
\end{defn}

The next question is how to extract the ``shape" of the point cloud from the simplicial complex. TDA defines this shape by examining the topological features present in this complex. More specifically, TDA extracts the topological features characterized by $k$-dimensional holes, where $k=0,1,2$ represents connected components, holes, and voids, respectively, and uses these features as a surrogate to describe the shape of the point cloud.

Another important question revolves around the choice of $\epsilon$, i.e., how do we define ``sufficiently close" points to construct a simplicial complex? The choice of $\epsilon$ is vital as a very large value can lead to just one connected component as all the points get connected to each other, hence a loss of any meaningful topological features. On the other hand, we might not see any interesting topological features when $\epsilon$ is too small due to no edge formation. TDA offers a systematic solution to address this question. Instead of making random choices of $\epsilon$ that could impact our conclusions, TDA involves analyzing the shape of the point cloud across various resolutions by studying the evolution of topological features as a function of $\epsilon$. To accomplish this, we construct a sequence of Vietoris–Rips simplicial complexes for $0=\epsilon_0<\epsilon_1<\ldots<\epsilon_n$ with  
$\{\epsilon_n\} \in \mathbb{R}^+\cup\{0\}$. We call this sequence of complexes, a Vietoris--Rips filtration and denote it by $\{\Rcal_{\epsilon_n}(X)\}_{n\in\mathbb{N}}$. 
Within this filtration, due to several topological transformations, such as the closure of existing holes or the emergence of new edges, several topological features would appear (birth) and disappear (death) in the corresponding simplicial complex. We finally assign each topological feature a `birth' and `death' value. The persistence (age) of these features provides a  meaningful characterization of the overall shape of the data; the features with longer persistence are considered more robust, significant, and representative of the shape of the point cloud. Conversely, features with shorter persistence are considered less significant or noise and can be filtered out if necessary.

Figures \ref{rips} depict this step-by-step procedure, also known as persistence homology, for obtaining the filtration from the point cloud. The idea is to encircle each point in the point cloud with $\epsilon$-sized balls, gradually expanding $\epsilon$ until all the balls merge, resulting in one connected component and keeping track of birth and death times of various topological features, such as connected components, loops, and voids.

\begin{figure}[t!]
		\begin{center}
			\subfigure[ $\epsilon_1$=0.2, loops=0]{%
				\includegraphics[height=4cm,width=6cm]{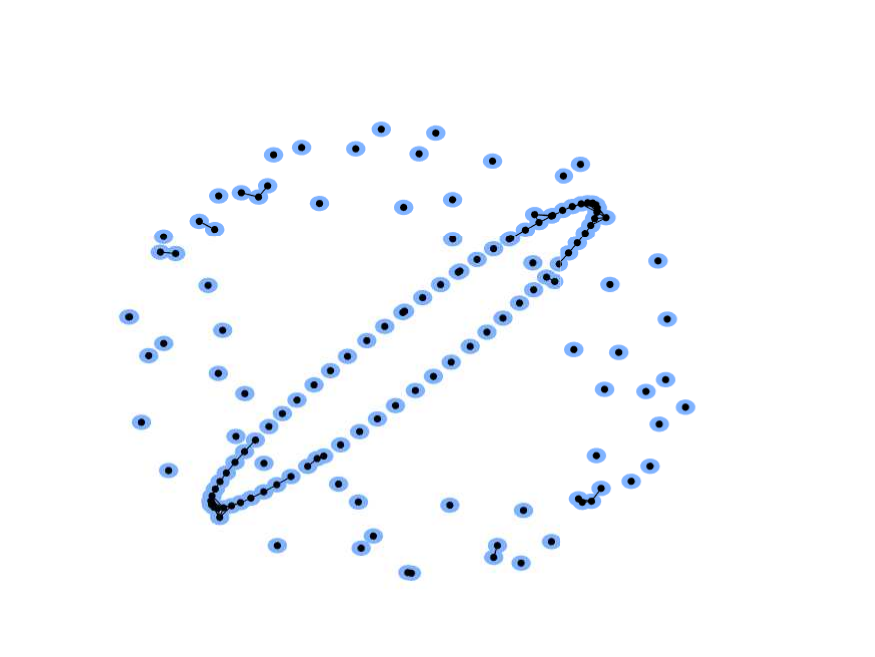}
				\label{rips1}}
			\quad
			\subfigure[ $\epsilon_3$=0.6, loops=1]{%
				\includegraphics[height=4cm,width=6cm]{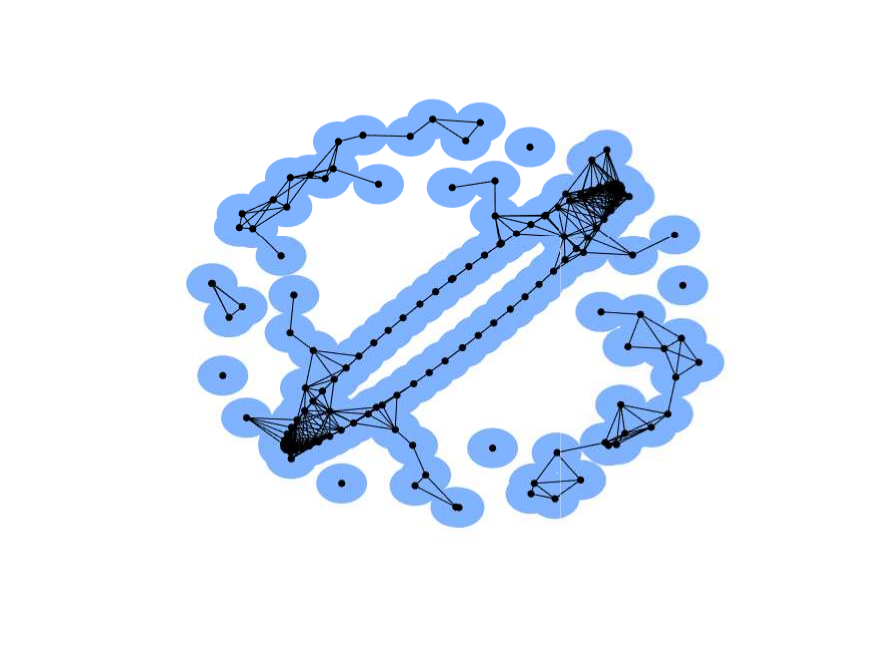}
				\label{rips2}}
			\quad
			\subfigure[ $\epsilon_5$=1.5, loops=2]{%
				\includegraphics[height=4cm,width=6cm]{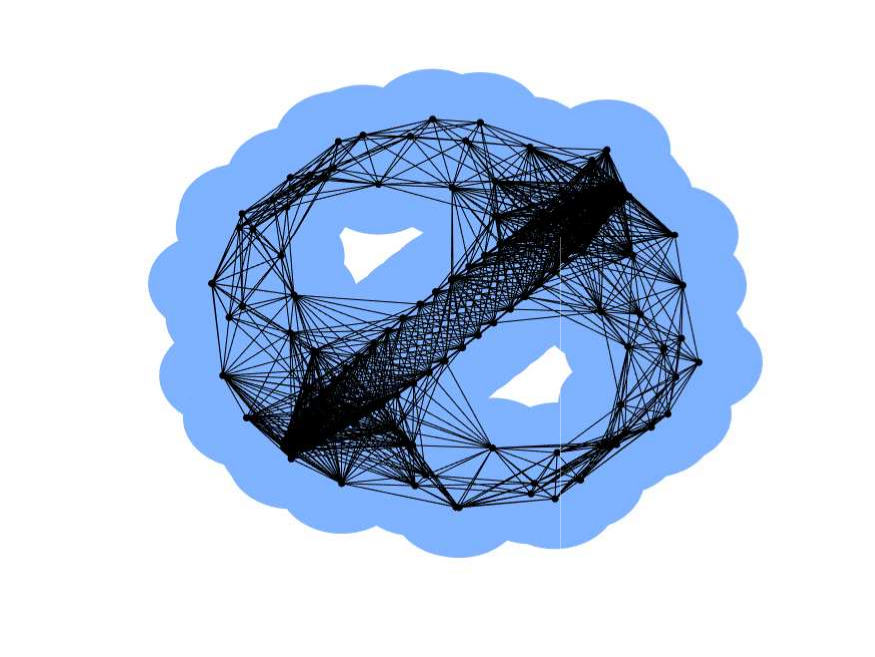}
				\label{rips33}}
			\quad
			\subfigure[ $\epsilon_6$=2.5, loops=0]{%
				\includegraphics[height=4cm,width=6cm]{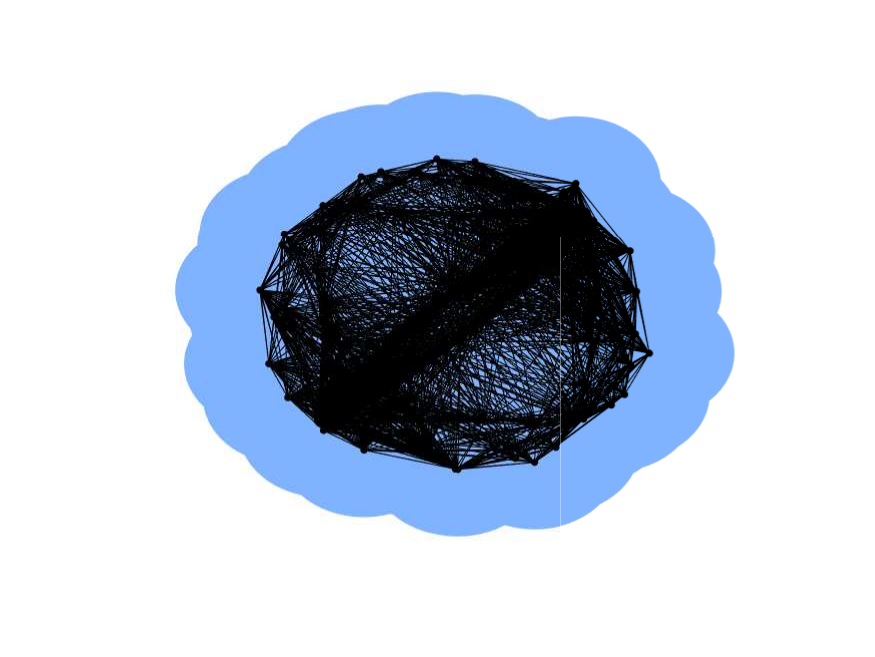}
				\label{rips44}}
		\end{center}
		\caption{ A diagram depicting the Rips filtration process. Each point in the point cloud is equipped with an equal-sized ball, with the radius of the ball serving as the filtering parameter. By increasing the value of the radius, a succession of nested simplicial complexes is formed. It causes characteristics like connected components and holes to emerge and vanish. The initial loop appeared in \ref{rips2} and died in \ref{rips33}. In \ref{rips33}, two additional substantial large holes can be observed, which die in \ref{rips44}.}\label{rips}
	\end{figure}
Having obtained the Rips filtration, we summarize the birth and death times of topological features in a persistence diagram (PD) $\Dcal$, which is a multi-set of points in $\mathbb{W} \times \{ 0,1,\ldots,q - 1\}$,  where $\mathbb{W} :=\{(b,d)\in \mathbb{R}^2 : d\geq b\geq 0\}$ and each element $(b,d,f)$ represents a topological feature of dimension $f$ that appeared at $\epsilon=b$ and disappeared at scale $\epsilon=d$ during the construction of filtration. Figure \ref{2a} gives the persistence diagram corresponding to the filtration in Figure \ref{rips}. When we study persistence diagrams for data analysis, one must be able to compare persistence diagrams. Therefore, it becomes necessary to introduce a distance metric between two persistence diagrams to describe their relation. The similarity between two persistence diagrams is commonly measured by $p$-Wasserstein distance, defined as

\begin{defn}
The $p$-Wasserstein distance \cite{cohen2010lipschitz} between two persistence diagrams $\Dcal_{ {X}}$ and $\Dcal_{ {Y}}$ corresponding to data sets $X$ and $Y$ is defined as 
\begin{equation*}
WD_p(\Dcal_{ {X}},\Dcal_{ {Y}})=\left( \inf_{\gamma} \sum_{a \in {\Dcal_{ {X}}}} ||a-\gamma(a) ||_{\infty}^p \right)^{\frac{1}{p}}. \end{equation*}
where $p\geq 1$ and $\gamma$ ranges over all bijections from the points in $\Dcal_{ {X}}$ to the points in $\Dcal_{ {Y}}$. For $p=\infty$, $p$-Wasserstein distance is called the bottleneck distance. 
\end{defn}
Though persistence diagrams contain potentially valuable information about the underlying data set, it is a multi-set which, when equipped with Wasserstein distance, forms an incomplete metric space and hence is not appropriate to apply tools from statistics and machine learning for processing topological features (\cite{bubenik2015statistical}). A persistence landscape is an alternative tool to convey the information in the persistence diagram. It comprises a sequence of continuous and piece-wise linear functions defined on a re-scaled birth-death coordinate, with the peaks of a persistence landscape representing the significant topological features. Given the birth-death pairs $p_i(a_i, b_i),~i=1,\ldots,m$, a persistence landscape, illustrated in Figure \ref{2b}, is defined as a sequence of functions $\lambda: \mathbb{N}\times\mathbb{R}\rightarrow [0,\infty)$, denoted as $\lambda(k, t) = \lambda_k(t)$. Here, $\lambda_k(t)$ represents the $k$-th largest value among the set $\{\Lambda _{p_i}(t) \mid i=1,\ldots,m\}$ and  
we set $\lambda_k(x) = 0$, ensuring that $\lambda_k(t) = 0$ for $k > m$ if the $k$-th largest value does not exist. Here, for a fixed feature dimension $f$, we associate a piece-wise linear function $\Lambda _p(t) : \mathbb{R}\rightarrow [0,\infty)$ with each birth-death pair $p = (a, b) \in \Dcal$, where $\Dcal$ is the persistence diagram, as follows:
\begin{eqnarray}
\Lambda _p(t) = \left\{\begin{array}{ll} t-a &\quad t \in [a,\frac{a+b}{2}] \\[0.5em]
b-t &\quad t \in [\frac{a+b}{2},b]\\[0.5em]
0 &\quad \mbox{otherwise}.
\end{array}\right. \label{pl}
\end{eqnarray}
Simply put, a persistence landscape is obtained by rotating the persistence diagram 45 degrees clockwise and then drawing right-angle isosceles triangles, treating the birth-death pair in the persistence diagram as a right-angle vertex. Finally, individual functions are traced out from this collection of right triangles, such that the first ($k$th) landscape function is the point-wise ($k$th) maximum of all the triangles drawn. 

One of the notable advantages of the persistence landscape representation, as compared to persistence diagrams (refer to \cite{bubenik2015statistical, bubenik2018persistence}), is that it forms a subset of the Banach space $L^p(\mathbb{N} \times \mathbb{R})$, which comprises sequences $\lambda= (\lambda_k)_{k \in \mathbb{N}}$. This set possesses an inherent vector space structure and becomes a Banach space when equipped with a norm that measures the degree of ``spread out" or ``concentration" exhibited by the landscape functions. The norm of a persistence landscape characterizes a point cloud by a single numerical value, and a higher norm value indicates the presence of more significant topological features in the underlying point cloud and is given by 
\begin{align}
	||\lambda||_p=\left( \displaystyle \sum_{k=1}^\infty ||\lambda_k||^p_p \right)^{\frac{1}{p}},
	\end{align}
where $||\cdot||_p$ is the $L^p$-norm.
 \begin{figure}[t!]
		\begin{center}
			\subfigure[Persistence diagram]{%
				\includegraphics[height=5cm,width=5cm]{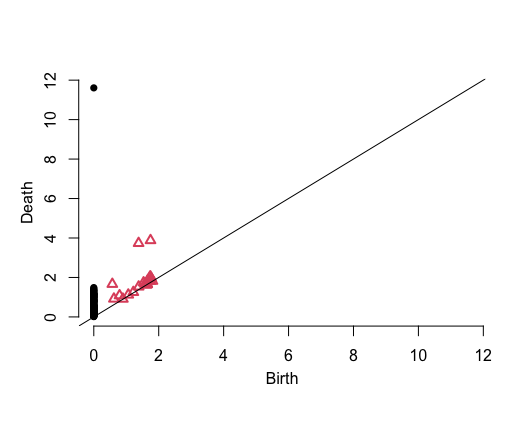}
				\label{2a}}
			\quad
			\subfigure[ Persistence Landscape ]{%
				\includegraphics[height=5cm,width=5cm]{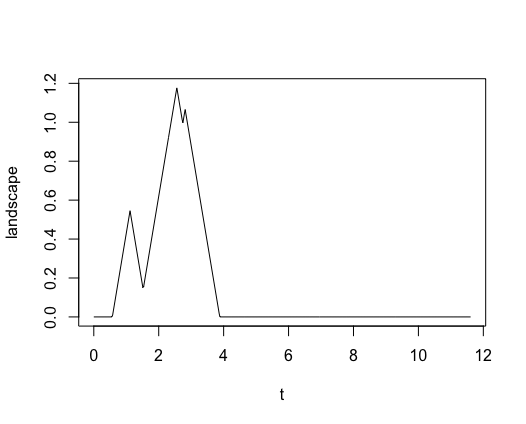}
				\label{2b}}
			\end{center}
		\caption{~ (a) The persistence diagram associated to Figure \ref{rips}; the
dots in it represent the birth and death of a feature; for instance, three off-diagonal red dots represent the birth and death of three significant holes in Figure \ref{rips}; ~(b) the corresponding persistence landscape.}\label{persistence}
	\end{figure}

Given that TDA operates on point clouds in multi-dimensional spaces, and our focus in this paper is on time series of stock returns, which lack natural point cloud representations, we must devise a method to obtain a point cloud from a time series. To achieve this, we adopt the standard procedure, called Taken's embedding (\cite{takens1981detecting}) that embeds a time series $x=\{x_t,t=1,\ldots, T\}$ to a multi-dimensional representation as follows:
\begin{figure}[ht!]
		\begin{center}
			\subfigure[ Time series]{%
\includegraphics[height=4cm,width=6cm]{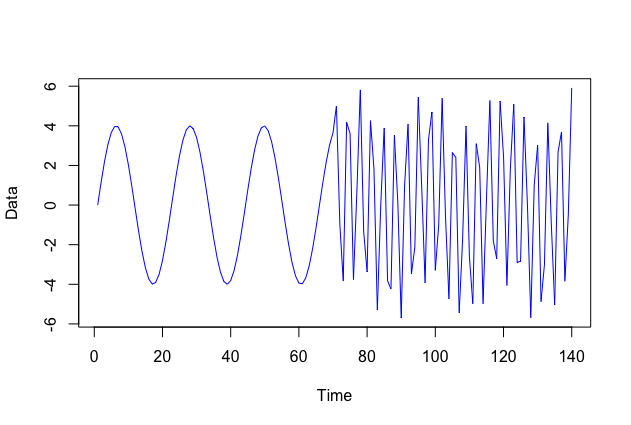}
				\label{1a}}
			\quad
			\subfigure[ Point cloud ]{%
\includegraphics[height=4cm,width=6cm]{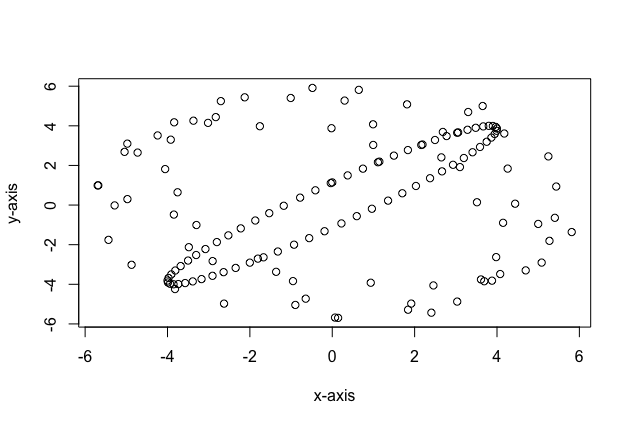}
				\label{1b}}
			\end{center}
		\caption{Time series and point cloud in $\mathbb{R}^2$ constructed using Takens' embedding}\label{takendata}
\end{figure}

    \begin{equation}
	{X} = \left[
	\begin{array}{c}
	X_{1}\\
	X_{2}\\
	\vdots\\
	X_{T-(d-1)\tau}\\
	\end{array}
	\right]=\left[
	\begin{array}{cccc}
	x_{1}& x_{1+\tau}\dots  &x_{1+(d-1)\tau}\\
	x_{2}& x_{2+\tau} \dots   &x_{2+(d-1)\tau}\\
	\vdots & \vdots & \vdots \\
	x_{T-(d-1)\tau} & x_{T-(d-2)\tau} \dots  &x_{T}\\
	\end{array}
	\right],\label{eq:point_cloud}
	\end{equation}
where $d$ is the dimension of reconstructed space (embedding dimension), $\tau$ is the time delay, and $T - (d - 1)\tau$ is the number of points in the new space. The row of the matrix $X$ denotes a point in the new space. We present the idea of Taken's embedding in Figure \ref{takendata}, where the first plot shows the time series. In contrast, the second plot presents the corresponding 2-dimensional point cloud, obtained using Taken's embedding with $d=2$ and $\tau=1$.

\section{New distance measures}\label{newdistances}

The notion of distance is inherently linked to the characteristics of the underlying data. For instance, the distance between two financial time series may differ from that between two DNA sequences. Nonetheless, distance measure aims to capture some sense of similarity, and often, one can employ the distance measure inherited from the reference space. However, if that is not the case, one must use a more suitable metric. For example, to compare persistence diagrams, Wasserstein distance (or bottleneck distance) is traditionally used due to their stability concerning perturbations of the input \cite{gidea2018topological,gidea2017topology,goel2020topological}. However, considering the time factor while dealing with financial time series is crucial, which the traditional Wasserstein distance fails to capture, as shown in the following toy example.

\textbf{Toy example 1}\label{toy}
Suppose we have two time series obtained from some experiment, $
   {x}=\{x_t,~t=1,\ldots,T\}$ and $  {y}=\{y_t,~t=1,\ldots,T\}$ in $\mathbb{R}^T$ such that $  {y}$ is the reverse of $  {x}$, that is, if $  {x}=(x_1,x_2,\ldots,x_T)$ then $  {y}=(x_T,x_{T-1},\ldots,x_1)$. The point clouds obtained using Takens' embedding theorem for $x$ and $y$ (say for $d=3$ and $\tau=2$) are
   \begin{align*}
 {X=\left[
	\begin{array}{ccc}
	x_{1}& x_{3}&   x_{5}\\
	x_{2}& x_{4}&   x_{6}\\
	\vdots & \vdots & \vdots \\
	x_{T-4}& x_{T-2}  &x_{T}\\
	\end{array}
	\right],Y=\left[
	\begin{array}{ccc}
    y_{1}& y_{3}&  y_{5}\\
	y_{2}& y_{4}&   y_{6}\\
	\vdots & \vdots & \vdots \\
	y_{T-4}& y_{T-2}&   y_{T}\\	
	\end{array}
	\right]
 =\left[
	\begin{array}{ccc}
x_{T}& x_{T-2}&  x_{T-4}\\
	x_{2}& x_{4}&   x_{6}\\
	\vdots & \vdots & \vdots \\
	x_{5}& x_{3}&   x_{1}\\	
	\end{array}
	\right].}
\end{align*}
since $y$ is reverse of $x$. Since TDA is a coordinate-free approach (or coordinate invariant), the point clouds for $x$ and $y$ will result in the same persistence diagrams and hence zero Wasserstein distance, implying a perfect similarity between the two, which is misleading. 


While the example above is hypothetical, it advocates for a more appropriate distance measure for time-series data. It emphasizes that Wasserstein distance might be beneficial when dealing with observations from a fixed geometrical object where one is trying to identify the structure of that object. Still, it is unsuitable for financial applications where time is essential. Further, the literature has observed that computing the Bottleneck or the Wasserstein distances is computationally demanding as it requires finding point correspondence as observed in \cite{agami2023comparison}. Therefore, we propose new distance measures between persistence diagrams and persistence landscapes that are efficient in terms of time and also consider the time component inherent in time-series data.

Consider two time-series, ${x}=\{x_t,~t=1,\ldots,T\}$ and ${y}=\{y_t,~t=1,\ldots,T\}$ in $\mathbb{R}^T$. For a given length $\ell\in\{1,\dots,T\}$, we define the sequence of $m$ sub-series
\begin{eqnarray*}
seq( {x})&=& \{(x_1,x_2,\ldots,x_{\ell})_1, (x_{1+\eta},x_{2+\eta}, \ldots,x_{{\ell}+\eta} )_2, \ldots,(x_{T-{\ell}+1}, x_{T-{\ell}+2},\ldots,x_T )_m \},
\end{eqnarray*}
where $\eta \in \{1,\ldots,{\ell}\}$ represents the degree of overlapping (for $\eta=\ell$ the sub-series are non overlapping and $\eta=1$ represents maximum overlapping), and $T-{\ell}=\eta(m-1)$. We now define distance between ${x}$ and ${y}$ using their corresponding sub-series as follows,

\begin{defn}\textbf{(Average $p$-Wasserstein distance ($AWD_p$))}
For $1 \leq p < \infty$, the average $p$-Wasserstein distance between the persistence diagrams $\Dcal_{X}$ of $  {x}$ and $\Dcal_{Y}$ of ${y}$ is defined as follows
\begin{equation*}
AWD_p(\Dcal_{X},\Dcal_{Y})= \sum_{i=1}^mw_i WD_p(\Dcal_{X_i},\Dcal_{Y_i})
\end{equation*}
where $WD_p(\Dcal_{X_i},\Dcal_{Y_i})$ represent the $p$-Wasserstein distance between the persistence diagrams $\Dcal_{X_i}$ and $\Dcal_{Y_i}$ corresponding to the point clouds of $i$th components in $seq( {x})$ and $seq({y})$. Here the weights $w_i$ are such that $w_i> 0$ and $\sum_{i}w_i=1$. For instance, one can simply take $w_i=\frac{1}{m}$. For $p=\infty$, we also call the average Wasserstein distance as the average bottleneck distance ($ABD$). 
\end{defn}

\begin{rem}
Computing $AWD_p$ is theoretically intuitive and computationally cheaper than calculating $WD_p$ between the two diagrams. Further, the choice of weights $w_i$ enables tuning of the influence of the preceding components of the time series when determining the distance.    
\end{rem}

\begin{defn}\textbf{(Differenced $p$-Wasserstein distance ($DWD_p$))}
The differenced $p$-Wasserstein distance ($DWD_p$) between the persistence diagrams $\Dcal_{X}$ and $\Dcal_{Y}$ of the point clouds ${ X}$ and ${Y}$ of the time series $  {x}$ and $  {y}$ is defined as
\begin{equation*}
 DWD_p(\Dcal_{X},\Dcal_{Y})= WD_p(\Dcal_{X- Y},\Dcal_{  0}), 
\end{equation*} 
 where $\Dcal_{X-Y}$ represents the persistence diagram corresponding to the point cloud obtained from the series $  {x}-  {y}$ and $\Dcal_{  0}$ is the diagram corresponding to a constant series (note that PDs are shift-invariant), containing only the diagonal \cite{mileyko2011probability}, so that $
WD_p(\Dcal_{X},\Dcal_{  0})=\left(2^{-p}\sum_{a\in \Dcal_{X}}(pers(a))^p\right)^{\frac{1}{p}}$, where $pers(a)$ denotes the persistence of $a=(b,d) \in \Dcal_x$, given by $pers(a)=d-b$, with $b$ and $d$ denoting the birth and death time.
\end{defn}

As highlighted above, the traditional Wasserstein distance, $WD_p(\Dcal_X,\Dcal_Y)$, has notable limitations. In the above toy-example, we observe that it successfully captures the topological differences between the time-series $x$ and $y$, but misses out to reflect the point-wise similarities as it is sensitive to phase shifts due to its dependence on the positioning of the topological features in $x$ and $y$. Additionally, computing $WD_p(\Dcal_X,\Dcal_Y)$ can be computationally intensive as it requires evaluating all bijections from the points in $\Dcal_{ {X}}$ to the points in $\Dcal_{ {Y}}$.

The average Wasserstein distance addresses this by utilizing a sliding window approach, dividing the time-series to multiple sub-series that allows to measure the changes in the shape of the point-cloud associated to each sub-series over the time-span of time-series. The robustness inherited in persistent homology against small perturbations makes average Wasserstein distance suitable primarily for analyzing noisy time-series. Further, it also facilitates capturing the abrupt shifts from one regime to other due to minor changes in external factors, called critical transitions, as demonstrated in \cite{akingbade2024topological}.

On the other hand, the differenced $p$-Wasserstein distance ($DWD_p$) aims to capture the point-wise discrepancies in the time-series' $x$ and $y$ over time. Clearly, a small value of $DWD_p$ implies that $x-y$ are close to zero implying strong topological similarity between $x$ and $y$. Moreover, it simplifies the computation of Wasserstein distance as it reduces the complexity by comparing a persistence diagram $\Dcal_{X-Y}$ to a constant persistence diagram. Furthermore, as it captures the point-wise similarity between $x$ and $y$, thus, it is more robust to small mis-alignments in the two time-series. This property, along with evaluation of similarity between $x$ and $y$ against a benchmark of zero, make it particularly useful for applications such as detecting anomalies or time-series matching.

Clearly, for the toy example, $AWD$ and $DWD$ are non-zero. A similar distance measures could also be defined on the space of persistence landscapes. Let $\lambda_x$ and $\lambda_y$ denote the persistence landscapes for the time series $x$ and $y$, respectively. The average $p$-Landscape distance ($ALD_p$) and the differenced $p$-Landscape distance ($DLD_p$) between time series $x$ and $y$ as follows
\begin{eqnarray}
\begin{aligned}
ALD_p(\lambda_x,\lambda_y)=\sum_{i=1}^mw_i\|\lambda^i_{  x}-\lambda^i_{  y}\|_p,~~ DLD_p(\lambda_x,\lambda_y)=\| {\lambda}_{  x-  y}\| _p, \\
\end{aligned}
\end{eqnarray}
where $\lambda_{x-y}$ is the persistence landscape corresponding to the point cloud of series $x-y$ and $\lambda_x^i$ and $\lambda_y^i$ are the persistence landscapes corresponding to the point clouds of $i$th components in $seq( {x})$ and $seq({y})$, respectively. 

\subsection{Stability of Persistence diagrams under new distance measures}

Having introduced new distance measures, it is essential to ensure that the mapping from data to a persistence diagram remains continuous as a data representation. Given that data is often noisy, it is crucial 
for the persistence diagrams to exhibit stability under perturbations in the data, that is, to demonstrate resilience to minor deformations in the data input. This section focuses on establishing the stability of persistence diagrams under the proposed distance measures.

\begin{defn}
Let $X$ and $Y$ be two compact sets in $\R^d$ i.e., $X,Y\subset \R^d$. The Hausdorff distance $HD(X,Y)$ between $X$ and $Y$ is defined as
\begin{equation}\label{defHd}
   HD(X,Y)=\inf\left\{\epsilon\ge 0\mid \text{$X\subseteq Y_\epsilon$ and $Y\subseteq X_\epsilon$}\right\},
\end{equation}
where $Y_\epsilon =\cup_{y\in Y} \{ z\in\R^d\mid \|z-y\|\le\epsilon\}$ denotes the generalized ball of radius $\epsilon$ around $Y$. Observe that
\begin{equation}\label{ABfinite}
  \text{$X\subseteq Y_\epsilon$ $\Longleftrightarrow$ for any $x\in X$ there exists some $y\in Y$ with $\|x-y\|\le \epsilon$.}
\end{equation}
\end{defn}
The following theorem relates the  Euclidean distance between two time series and Hausdorff distance between their corresponding point clouds.

\begin{theorem} \label{thm3}
Let $  {x}$ and $  {y}$ be two time series in $\R^T$. Let $  X=\{ x_{I_1},\dots,x_{I_m}\}$ and $  Y=\{ y_{I_1},\dots,y_{I_m}\}$ be their Takens' embeddings in $\R^d$. Then 
\begin{equation}\label{tempo}
HD(  X,  Y) \le \|   {x}-  {y}\|_2.
\end{equation}
\end{theorem}
\begin{proof}
For any $j=1,\dots,m$, we have
$\| x_{I_j} - y_{I_j}\|^2_2 =\sum_{t\in I_j} (x_t-y_t)^2 \le \sum_{t=1}^T (x_t-y_t)^2=\|   {x}  -   {y} \|^2_2.$ In view of \eqref{ABfinite}, we infer that $X\subseteq Y_{\|   {x}  -   {y} \|}$ and $Y\subseteq X_{\|   {x}  -   {y} \|}$. Combining this with \eqref{defHd}, we arrive at Equation (\ref{tempo}).
\end{proof}
We state the propositions required for our main theorem.
\begin{prop} \label{prop1} (\cite{cohen2007stability},\cite{kusano2017kernel},\cite{chazal2014persistence}) Let $(E,d_E)$ be a metric space and $  X$ and $  Y$ be finite subsets of $E$.  Then the persistence diagrams $\Dcal_{X}$ and $\Dcal_{Y}$ satisfy
\begin{equation*}
BD(\Dcal_{X},\Dcal_{Y}) \leq HD(  X,  Y). \end{equation*}
\end{prop}
\begin{prop} \label{p3}
\cite{cohen2010lipschitz},\cite{kusano2017kernel} Let $1 \leq \hat{p} \leq p < \infty$, and $\Dcal_{X}$ and $\Dcal_{Y}$ be persistence diagrams whose degree $p$ total persistences\footnote{Given a persistence diagram $\Dcal_X=\{a_j=(b_j,d_j)\}$, the degree $p$ total persistence of $\Dcal_X$ is defined as sum of persistence of all the points of degree $p$ in $\Dcal_X$ i.e., $Pers_p(\Dcal_X)=\sum_j (d_j-b_j)$. } are bounded from above. Then,
{\footnotesize
\begin{equation*}
WD_p(\Dcal_{X},\Dcal_{Y}) \leq \left( \frac{Pers_{\hat{p}}(\Dcal_{X})+Pers_{\hat{p}}(\Dcal_{Y})}{2}\right )^{\frac{1}{p}} BD(\Dcal_{X},\Dcal_{Y})^{1-\frac{\hat{p}}{p}}, 
\end{equation*}}
where $Pers_{\hat{p}}(\Dcal_{X}):= \sum_{a \in \Dcal_{X}}pers(x)^{\hat{p}}$ for $1 \leq \hat{p} < \infty$.
\end{prop}

\begin{prop} \label{p4} \cite{kusano2017kernel}
Let $E$ be a triangulable compact subspace in $\R^d$, $ {  X}$ be a finite subset of $E$, and $p > d$. Then,
$$ Pers_{p}(\Dcal_{X}) \leq \frac{p}{p-d}C_E diam(E)^{p-d},  $$
where $C_E$ is a constant depending only on $E$.
\end{prop}


\begin{theorem} \label{main_result}
Let ${x}$ and ${y}$ be two time series in $\R^T$. Let $  X=\{ x_{I_1},\dots,x_{I_m}\}$ and $  Y=\{ y_{I_1},\dots,y_{I_m}\}$ be their Takens' embeddings in $\R^d$. Then 
 \begin{equation} \label{abd_1}
     ABD(\Dcal_{X},\Dcal_{Y})  \le \| {x}-{y}\|_2,
 \end{equation}
 \begin{equation} \label{awd_1}
   AWD_p(\Dcal_{X},\Dcal_{Y})  \le \left(\frac{\hat{p}}{\hat{p}-d}C_E diam(E)^{\hat{p}-d}\right)^\frac{1}{p} \| {x}-{y}\|_2^{1-\frac{\hat{p}}{p}}.
   \end{equation}
\end{theorem}
\begin{proof}
By definition, for any choice of $\eta$ and $l$, $X_i,~Y_i;~ i \in \{1,2,\ldots,m\}$, which represents the point cloud of $i$th component in $seq( {x})$ and $seq( {y})$, say $ {x}_i$ and $ {y}_i$, respectively, are again subsets of $E$. Therefore, using Proposition \ref{prop1}, we have
\begin{eqnarray*}
BD(\Dcal_{X_i},\Dcal_{Y_i}) &\leq& HD(  X_i,  Y_i)~~ \forall ~~i\\
\implies 
\sum_{i=1}^mw_i BD(\Dcal_{X_i},\Dcal_{Y_i}) &\leq& \sum_{i=1}^mw_i HD(  X_i,  Y_i),~~(\because ~ w_i>0, ~ \sum_{i=1}^{m}w_i=1)\\
& \leq& \sum_{i=1}^mw_i \|  {x}_i  -  {y}_i\|_2, ~~ (\mbox{using Theorem \ref{thm3}})\\
&\leq& \sum_{i=1}^mw_i \|  {x}  -  {y}\|_2, ~~(\mbox{by definition of Euclidean distance})\\
&\leq&  \|  {x}  -  {y}\|_2 \sum_{i=1}^mw_i
\leq  \|  {x}  -  {y}\|_2 ~~ (\because \sum_{i=1}^mw_i=1.)
\end{eqnarray*}
Therefore, we have (\ref{abd_1}). Again, by definition, we have
\begin{equation*}
AWD_p(\Dcal_{X},\Dcal_{Y})= \sum_{i=1}^mw_i WD_p(\Dcal_{X_i},\Dcal_{Y_i})
\end{equation*}
Using Proposition \ref{p3} and \ref{p4}, we have
{\footnotesize
\begin{eqnarray*} 
\sum_{i=1}^mw_i WD_p(\Dcal_{X_i},\Dcal_{Y_i})
&\leq&\sum_{i=1}^mw_i \left( \frac{Pers_{\hat{p}}(\Dcal_{X_i})+Pers_{\hat{p}}(\Dcal_{Y_i})}{2}\right )^\frac{1}{p} BD(\Dcal_{X_i},\Dcal_{Y_i})^{1-\frac{\hat{p}}{p}}\\
 &\leq& \left(\frac{\hat{p}}{\hat{p}-d}C_E diam(E)^{\hat{p}-d}\right)^\frac{1}{p}\sum_{i=1}^mw_iBD(\Dcal_{X_i},\Dcal_{Y_i})^{1-\frac{\hat{p}}{p}}\\
 &\leq& \left(\frac{\hat{p}}{\hat{p}-d}C_E diam(E)^{\hat{p}-d}\right)^\frac{1}{p}\sum_{i=1}^mw_iHD({X_i},{Y_i})^{1-\frac{\hat{p}}{p}}\\
  &\leq& \left(\frac{\hat{p}}{\hat{p}-d}C_E diam(E)^{\hat{p}-d}\right)^\frac{1}{p}\|  {x}- {y}\|_2^{1-\frac{\hat{p}}{p}}.
 \end{eqnarray*}}

Therefore, we have result (\ref{awd_1}).
\end{proof}

\begin{rem}
The $DWD_p$ is a special case of $WD_p$ with one component as zero $PD$. Therefore, the stability of persistence diagrams with respect to $DWD_p$ is trivial. Further, the stability of the proposed distances on persistence landscapes follows directly from the landscape stability theorems, Theorem 12 - Theorem 16 in \cite{bubenik2015statistical}. Further, the distance metrics $AWD_p,	DWD_p, ALD_p$ and $DLD_p$ are pseudo-metrics. The properties of non-negativity and symmetry are easily verified. Also, the triangular inequality property is satisfied since they are defined using the Wasserstein and $p$-landscape distance, which are again pseudo-metrics.
\end{rem}


\section{Our Methodology for Sparse Portfolios}\label{our_methodology}
Consider a stock-market index $I$ consisting of assets $i\in\Ical=\{1,2,\ldots,n\}$, in discrete time $t=0,1,2,\ldots,$, where $t$ represents the end of a trading day. At any time $t$, we denote the price of the index $I$ and any asset $i\in \Ical$ by $P_{0,t}$ and $P_{i,t}$ respectively. Further, we denote by $R_{i,t}$ and $r_{i,t}$, the simple-return and log-return of the asset $i$ over the period $[t-1,t]$ with $t=1,2,\ldots,$ and are obtained as follows,
\begin{equation*}
R_{i,t}=\frac{P_{i,t}-P_{i,t-1}}{P_{i,t-1}},~~ r_{i,t}=\ln\,\frac{P_{i,t}}{P_{i,t-1}},\;\; i \in \Ical\cup\{0\};
\end{equation*}
Now, consider an investor at time $T$, who desires to create a portfolio consisting of assets in $\Ical$, based on the historical observations, i.e., the time-series of prices (and hence of returns), $(P_{0,t}, P_{1,t},\ldots,P_{n,t}),~t=0,1,2,\ldots,T$. We define a portfolio as an $n\times 1$ vector, denoted by $\bm w=(w_{1},w_{2},\ldots,w_{n})$, where $w_{i},i\in\Ical$ denotes the proportion of investment in the $i$th asset. We assume that the investor follows a so-called buy-and-hold strategy, i.e., she plans to hold this portfolio until the end of a specific investment horizon. Further, the set of all admissible portfolios is denoted by $\Wcal$, 
\begin{equation*}
\Wcal=\left\{{\bm w}=(w_1,w_2,\ldots,w_n): \sum_{i\in \Ical} w_i=1,w_i\geq 0~~\forall i\in \Ical \right\}.
\end{equation*}


There are multiple ways to select the portfolio $\bm w$ depending on the goals and risk appetite of the investor. We discuss three of them, namely, mean-variance ($MV$) optimal portfolio, global minimum variance ($GMV$) portfolio, and index tracking ($IT$) portfolio with budget and short selling constraints.

\textbf{Mean-Variance Portfolio} In the standard Markowitz, i.e., the mean-variance portfolio optimization problem, the goal is to select a portfolio that delivers the highest expected returns for a given variance. Mathematically, the $MV$ portfolio corresponds to the solution of the following optimization problem, 
	\begin{align} (MV) \qquad & \max_{\bm w\in \Wcal}
	\begin{aligned}[t]
	\bm w'\bm\mu-\frac{\gamma}{2}\bm w'\bm\Sigma \bm w \label{mvp}
	\end{aligned} 
	\end{align} 
where $\bm\mu=(\mu_1, \ldots, \mu_n) $ is the sample mean vector and $\bm\Sigma$ is the sample variance-covariance matrix of the asset returns calculated over the observation period, $t=0,1,2\ldots,T$.  Further, $\bm w'\bm \mu$ and $\bm w'\bm\Sigma \bm w$ denote the return and the variance of the portfolio. The parameter $\gamma$ is the risk aversion parameter that controls the trade-off between the portfolio’s risk and return. We follow \cite{demiguel2009optimal} and set $\gamma =1$ for the empirical analysis. 

It is well documented in the literature that estimating mean vector $\bm \mu$ is more difficult than the variance-covariance matrix $\bm\Sigma$ of asset returns. Further, the impact of errors in $\bm \mu$ estimates is significantly larger than that of $\bm \Sigma$ on the portfolio $\bm w$. Therefore, we also consider a global minimum variance portfolio, which relies solely on estimates of $\Sigma$ and thus is less vulnerable to estimation error than $MV$ portfolio. The $GMV$ portfolio is the solution to the following optimization problem,
	\begin{align} (GV) \qquad & \min_{\bm w \in \Wcal}
	\begin{aligned}[t]
\bm w'\bm \Sigma \bm w \label{gvp}
	\end{aligned} \end{align} 
\textbf{Index Tracking} In an index tracking portfolio, the goal is to select the portfolio $\bm w$ that aims to mimic the performance of the index $I$ by minimizing the deviation of the portfolio return from the index return, the measure commonly known as tracking error $TE(\bm w)$. The problem can be formulated as 
\begin{align} \qquad & \min_{\bm w \in \Wcal}
	\begin{aligned}[t]
TE({\bm w }) \label{tev1}
	\end{aligned}  
	\end{align} 
The most common choice for the tracking error $TE(\bm w)$ in the literature is the mean of the squared deviations of the portfolio return and the index return, traditionally called the tracking error variance ($TEV$). We also adopt $TEV$ as the tracking error measure \footnote{One of the key benefits of employing this metric for tracking error is that the objective function is quadratic in the asset weights, allowing for using standard and very efficient quadratic programming algorithms to minimize it. A further advantage of this tracking error is that it avoids the issues associated with metrics based solely on the variance of the deviation since these measures ignore the possibility of a constant difference between the two returns (\cite{beasley2003evolutionary}). Instead of the squared deviations, the absolute error is also a popular choice in the literature (e.g., \cite{prigent2007portfolio,rudolf1999linear}). This is advantageous because the optimization problem reduces to a linear program that can be efficiently solved. In practice, squared deviation is preferred over absolute errors as a statistical criterion, especially when large errors are undesirable.} which is defined as 
\begin{equation*}
\frac{1}{T}	\sum_{t=1}^{T}(\sum_{i =1}^n w_{i}R_{i,t}-R_{0,t})^2,
\end{equation*}
where $R_t(\bm w)=\sum_{i =1}^n w_{i}R_{i,t}$ denotes the portfolio return at time $t$.

\subsection{TDA-based Sparse Portfolios}

As discussed in the Introduction, various approaches for asset selection have been documented in the literature, especially for index-tracking and the Markowitz model to obtain sparse portfolios. In this paper, we adopt a strategy based on clustering that involves utilizing the distance measures introduced on persistence diagrams and landscapes. Building upon the new topology-based similarity measures, we seek to organize assets or clusters so that assets in the same cluster are similar and assets across the clusters are dissimilar in a topological sense. Once the clustering procedure is completed, we select the assets based on the following two strategies:  

\textbf{Strategy 1:} This strategy corresponds to selecting assets that share the cluster with the index, intending to replicate the index's performance. This choice is intuitive because the assets in the same cluster have a similar topological structure as quantified by the similarity matrix, thus narrowing down the number of assets potentially more suitable to mimic index performance. We then solve the tracking error optimization problem (\ref{tev1}) on the selected assets to find the optimal weights.

\textbf{Strategy 2:} This strategy involves selecting assets for mean-variance and global minimum-variance portfolios. The underlying idea is to choose the representative assets from each of the clusters, driven by the intuition that this strategy promotes sparsity due to a smaller number of selected assets and diversity due to inter-cluster dissimilarity, as quantified by the similarity matrix. Once the assets are selected, we solve (\ref{mvp}) and (\ref{gvp}) to find the optimal weights. In this strategy, one needs to select the cluster representative. One possible way is to choose the representative based on criteria, such as the asset with the highest Sharpe ratio or the minimum standard deviation. However, this leads to a selection bias. Therefore, instead of imposing some random selection criteria, we propose selecting the cluster's exemplar (representative center).

In this paper, we employ affinity propagation clustering (APC) (\cite{redmond2011affinity}) to obtain clusters.\footnote{The details on how APC works is presented in Appendix A.} Our choice of adopting APC mainly stems from the requirement that the representative center of the clusters must be a part of the data set (refer to Strategy 2). Further, we should get the same result, i.e., the same clusters and centers, if we run our strategy on the same data set but on different platforms/software. Clustering techniques, such as $k$-means and spectral clustering, assign random data points as initial cluster centers and shift these centers in each iteration. The drawback of these methods is that they yield different clustering sets upon each execution due to random initialization, and the centroids often lack a meaningful connection to actual variables in the dataset (\cite{jain1999data}). Affinity propagation addresses this issue by directly returning exemplars that correspond to actual data points.

Moreover, APC demonstrates superior accuracy compared to other clustering algorithms (\cite{redmond2011affinity}) such as $k$-means (\cite{macqueen1967some}), and spectral clustering (\cite{shi2000normalized,murzin1995scop}). Notably, APC doesn't require the pre-specification of the number of clusters, unlike $k$-means and spectral clustering. Consequently, the similarity matrix is the only input needed for affinity propagation clustering. However, it requires the choice of $\lambda$ and $P$ (refer Appendix A for details on these parameters). We chose P as the median of the similarity matrix, a common choice in the literature, as it yields a moderate number of clusters. We run the clustering algorithm for various values of $\lambda$ and select the optimal $\lambda$ based on the Silhouette score of clusters. Finally, as we deal with time series of stock returns, which typically behaves like white noise, clustering algorithms such as $k$-means would not be suitable as demonstrated in \cite{zheng2020diversity}.\footnote{We also compare the portfolio performance from APC to that from $K-$medoids and Hierarchical clustering in the Section \ref{appendix_comparison}}


Given the new topology-based distance measures, we define the following similarity matrices for the set of time series $x,y \in \mathbb{R}^T$ to be used as an input in affinity propagation clustering,
\begin{alignat}{2}
    K_1(x, y) & = e^{-\frac{AWD_p(\Dcal_{x}, \Dcal_{y})^2}{\sigma^2}}, \quad &
    K_2(x, y) & = e^{-\frac{Wass_p(\Dcal_{x-y}, \Dcal_{0})^2}{\sigma^2}}, \notag \\
    K_3(x, y) & = e^{-\frac{ALD_p(\lambda_x, \lambda_y)^2}{\sigma^2}}, \quad &
    K_4(x, y) & = e^{-\frac{DLD_p(\lambda_x, \lambda_y)^2}{\sigma^2}},
\end{alignat}
where $\sigma$ is a free parameter.\footnote{These are kernels in strict sense \cite{paulsen2016introduction} and induce similarity matrices which are semi-positive definite and hence can be used as input for the clustering algorithms.} It is worth highlighting that apart from the advantage of the new TDA-based distance measure over $p$-Wasserstein distance that it allows to compare time-series data as discussed in Section \ref{newdistances}, it offers a substantial computational speedup as it reduces the computational complexity to compute the persistence landscape for a large time series by dividing it into computations over several sub-series. Further, this division into several sub-series accounts for heteroscedasticity and non-stationarity in the financial data.

For comparison purposes, we also utilize the following kernels based on traditional correlation measures, namely Spearman’s rank correlation and Pearson's correlation,
\begin{eqnarray}
\begin{aligned}
K_5(x, y)=e^{-\frac{d_1(  x,  y)^2}{\sigma^2}},~
K_6(x, y)=e^{-\frac{d_2(  x,  y)^2}{\sigma^2}},\\
\end{aligned}
\end{eqnarray}
where $d_i(x,y) =\sqrt{2(1-\rho_i(x,y))}~~i\in\{1,2\}$, with $\rho_1$ and $\rho_2$ denoting the Spearman’s rank correlation and Pearson's correlation, respectively.\footnote{The distance measure $d_1$ defined using Pearson's correlation is the most commonly used distance measure in the literature \cite{onnela2003dynamics}. However, because these $ x, y $ are logarithmic returns in our scenario, Spearman's or Kendall's rank correlation coefficient is a much better choice due to its robustness \cite{zheng2020diversity}. As a result, we select both of these correlation measures for comparison.} We use the same notation $K_i~i=1,\ldots,6$ to represent the kernels and the corresponding similarity matrix.

The kernels, $K_1$ to $K_6$, depend on hyper-parameter $\sigma$, which needs to be estimated from training data. This paper adopts a data-driven approach for estimating $\sigma$, specifically employing the local scaling algorithm \cite{yang2008self, shang2012fast}. More precisely, we use a local scaling parameter $\sigma_{ij}$ for each pair of input data $(x_i,x_j)$ defined as
\begin{equation*}
\sigma_{ij}= d(x_i, x_m(x_i)) d(x_j, x_m(x_j)), 
\end{equation*}
where $x_m(x_\cdot)$ denotes the $m$th nearest neighbor of the input $x_\cdot$ in terms of distance measured by $d$ used in the corresponding kernel.\footnote{The selection of $m$ is independent of scale. We take $m=7$ in our analysis since it provides good clustering results \cite{yang2008self}.} The adoption of diverse scaling parameters for each pair $(x_i,x_j)$ rather than a single scaling parameter $\sigma$ allows us to exploit the local statistics of the neighborhoods around the corresponding input pair, thereby facilitating a data-driven self-tuning of $\sigma$.

\subsection{Performance Measures}

Our goal is to evaluate the out-of-sample performance of the portfolio strategies mentioned above: Strategy 1 for index tracking and Strategy 2 for mean-variance and minimum variance optimal portfolios. We study the performance of these portfolios across various standard datasets in the literature on portfolio construction (more details on the datasets follow in the next section). 

Given a daily time series of out-of-sample returns of optimal index tracking portfolio from strategy 1, we compute the five metrics. First, we calculate out-of-sample tracking error (TE): $\text{TE}={\frac{\sum_{t=1}^{T}(R_t(\bm w)-R_{0,t})^2}{T}}.$

It measures how closely the optimally tracking portfolio replicates the benchmark index. Further, we also apply one-tailed $t$-test, with null hypothesis $H_0:\, \text{TE}_{s_1}-\text{TE}_{s_2} \leq 0$ against the alternative hypothesis $H_a: \text{TE}_{s_1}-\text{TE}_{s_2} > 0$, to test whether the TE from strategy $s_1$ is statistically indistinguishable to that from strategy $s_2$.{\footnote{ The $t$-test is used to test if the mean $\mu$ of the population is ($\geq$ ~or~ $\leq$) to some constant $\mu_0$ when the variance of the population is unknown. To test the following hypothesis
\begin{equation*}
H_0: \mu\leq \mu_0~~~\mbox{against}~~~H_a: \mu> \mu_0,
\end{equation*}
the test statistics is given by
$t=\frac{\bar{x}-\mu_0}{\frac{S}{\sqrt{N}}},$ where $S$ and $\bar{x}$ denote the sample variance and mean, respectively, and $N$ is the
sample size. The critical region for the test is $t>t_{\alpha,N-1},$
where $\alpha$ is the significance level, and $t_{\alpha,N-1}$ denotes
the critical value of the $t$-distribution with $N-1$ degrees of
freedom. The $p$-value can be obtained as $P(t_{N-1}>t)$ and we
reject the null hypothesis if $p-$value is less than $\alpha$.}} Second, we calculate excess mean return (EMR), defined as the mean of the difference between the out-of-sample returns from the optimal tracking portfolio and the returns of the benchmark index, that is, $\text{EMR}= \frac{\sum_{t=1}^{T}(R_t(\bm w)-R_{0,t})}{T}.$ A positive (negative) EMR indicates that the tracking portfolio outperforms (underperforms) the benchmark index. 

Third, to measure the extent to which the optimal tracking portfolio returns levels differ from that of the benchmark, we calculate Pearson's correlation coefficient (COR) between the returns from the tracking portfolio and the benchmark index. A strategy for index tracking is appealing if its correlation is close to one. Fourth, we calculate the information ratio (IR), defined as the ratio of excess mean returns and tracking error, to measure the trade-off between excess mean return and tracking error: $\text{IR}=\frac{\text{EMR}}{\text{TE}}.$
Higher values of the information ratio are desirable. 

Finally, to measure the tracking portfolio's stability and understand the average wealth traded while holding (and re-balancing) the portfolio, we compute the portfolio turnover (TR). Formally, it is defined as the average absolute value of trades among $n$ stocks:
$\text{TR}=\frac{1}{N} \sum _{r=1}^{N}\sum_{i=1}^{n}|w^r_{i}-w^{r-1}_{i}|,$
where $N$ is the total number of re-balancing windows, and $w^r_{i}$ and $w^{r-1}_i$ are respectively the weights of $i$th asset in $r$th and $(r-1)$th window. One can interpret $TR$ as a measure of the transaction costs for implementing the corresponding portfolio strategy. For instance, a lower TR means a lower transaction cost.
 
Similarly, given the time series of daily out-of-sample returns obtained from mean-variance and minimum variance optimal portfolios, we follow \cite{demiguel2009optimal}
and compute four quantities to assess the performance of the portfolios. First, we measure the out-of-sample mean, which is given by $\mu= {\sum_{t=1}^{T}\frac{R_t(\bm w)}{T}}.$ To test whether the mean from the portfolio is statistically distinguishable, we also compute the p-value of the difference using the $t$-test.
We apply a one-tailed $t$-test on null
hypothesis $H_0:\, \mu_{s_1}-\mu_{s_1} = 0$ against the alternative
hypothesis $H_a: \mu_{s_1}-\mu_{s_1} > 0$ to test if the out-of-sample mean return
from $s_1$ is statistically significant from, $s_2$.

Second, we calculate the out-of-sample Sharpe ratio (SR), which is defined as the ratio of the sample mean to the sample standard deviation, i.e., $\text{SR}= \frac{\mu}{\sigma}$. Furthermore, we test whether the out-of-sample Sharpe ratios of two strategies, $s_1$ and $s_2$, are statistically different, we use the hypothesis test suggested by \cite{jobson1981performance} after adjusting for the correction highlighted
in \cite{memmel2003performance}.
We apply the one-sided $z_{\text{SR}}$ test with the hypothesis $H_0:\, \text{SR}_{s_1} - \text{SR}_{s_2} = 0$ and $H_a: \text{SR}_{s_1} - \text{SR}_{s_2} > 0$.\footnote{Given two strategies $s_1$ and $s_2$, with $\mu_{s_1}$, $\mu_{s_2}$, $\sigma_{s_1}$, $\sigma_{s_2}$, $\sigma_{s_1,s_2}$ as their sample means, standard deviations, and the covariance of two strategies over a sample period $n$. The $z$-test statistic is~$$z_{\text{SR}} = \frac{\sigma_{s_2}\mu_{s_1} - \sigma_{s_1}\mu_{s_2}}{\sqrt{\Upsilon}},~~ \mbox{with} $$ 
\begin{eqnarray*}
\Upsilon = \frac{1}{n}\big(2\sigma_{s_1}^2\sigma_{s_2}^2 - 2\sigma_{s_1}\sigma_{s_2}\sigma_{s_1,s_2} + 0.5\mu_{s_1}\sigma_{s_2}^2 + 0.5\mu_{s_2}\sigma_{s_1}^2
- \frac{\mu_{s_1}\mu_{s_2}}{\sigma_{s_1}\sigma_{s_2}}\sigma_{s_1,s_2}^2\big).\end{eqnarray*}
}
Along with the above performance measures, we calculate the certainty-equivalent (CEQ) return, defined as the risk-free rate an investor is willing to accept rather than adopting a particular risky portfolio strategy. Formally, we compute the CEQ return as $\text{CEQ}=\mu-\frac{\gamma}{2} \sigma^2,$
where $\sigma^2$ is the variance of the out-of-sample returns of the portfolio. We take $\gamma =1 $ as in \cite{demiguel2009optimal} in the analysis. We further test for the statistical significance of the difference of CEQ returns from two strategies \cite{demiguel2009optimal}. Fourth, we calculate the portfolio turnover ratio (TR) as a proxy for the transaction cost.

\section{Empirical Results}\label{empirical_results}

This section presents our empirical analysis. Section \ref{data} describes the data used in our study and introduces the rolling window scheme employed in the analysis. Sections \ref{sparse_index_tracking} and \ref{mean_variance} contain the performance results of our sparse index tracking and sparse mean-variance portfolios, respectively. Additionally, within these sections, we further validate the robustness of our findings through a performance analysis on the recent COVID-19 data.

\subsection{Data}\label{data}

We consider the daily closing prices of constituents of the S\&P 500 index, obtained from the Thomson Reuters datastream spanning a period from 1st Dec 2009 to 31st August 2022.\footnote{We chose the constituents of the S\&P 500 index as it includes the largest and the most liquid companies in the US markets and hence is widely used as a benchmark by both academics and practitioners alike. Further, the S\&P 500 index is relatively stable and has least distortions driven by idiosyncrasies of small and less liquid stocks, whose behavior is not a representative of market-wide behavior. Thus, it provides a reliable dataset over a long duration that has minimum distortion due to less frequent re-balancing.  Moreover, it provides a robust proxy to analyze economic trends and broader dynamics of the markets due to its comprehensive coverage of large as well as well-established companies spanning various sectors. At last, institutional investors closely monitor the performance of S\&P 500 constituents, thus, making the index a key indicator of market sentiments and risk perceptions, particularly during major financial events. For completeness, we have also tested the performance of our TDA-based sparse index tracking portfolios on the benchmark data sets from \cite{beasley2003evolutionary} and \cite{beas2009} for index tracking problems (refer Table \ref{bea}). These data sets are drawn from several major world equity indices, and thus, analyzing the performance over these data sets gives additional evidence on the adequacy and effectiveness of the TDA-based portfolios.} The selection of the period stems from various reasons. First, we chose the data period according to the availability of a maximum number of constituents, resulting in 451 stocks in our dataset. Second, it covers a range of particularly significant market conditions, starting from the recovery after the sub-prime crisis to the onset of unprecedented COVID-19 pandemic, when the global markets experienced an unmatched turmoil, followed by a historic rebound to pre-crisis peaks and a stable period afterwards. Thus, the selected period allows to assess the performance of TDA across a wide range of economic regimes, including both normal and highly volatile scenarios. The COVID-19 crisis period, in particular, provides a perfect environment to test the robustness of our approach under extreme volatility.

We divide the study period into three non-overlapping periods. The first period covers a relatively stable period starting in December 2009 to July 2019. The second period, August 2019 to August 2020, allows to evaluate the robustness of TDA based approach during a crucial time of COVID-19 pandemic. The remaining period, September 2020 to August 2022 is utilized for comparison purposes, in particular, to other clustering approaches. There are several reasons for selecting August 2019 to August 2020 for COVID-19 analysis. First, choosing August 2019 as the starting month allows us to use February 2020 as the first test month. Second, it covers an extraordinary financial market event, where markets across the globe faced an unmatched turmoil, resulting in a sharp decline of 34\% of the S\&P500 from its pre-crisis peak in February 2020. This is in contrast to a rather slow paced decline during the sub-prime crisis. Third, as pointed out in \cite{adrian2023macro}, markets encountered a historic rebound to pre-crisis levels by mid-August 2020, owing to the policy responses such as liquidity support, cuts in the interest rates, emergency credit availability, thus, showing one of the fastest recoveries. Therefore, the selected period allows us to assess the models' performance across various phases of the market, including extreme volatility and sharp recovery, thus offering insights into models' performance under stressed scenarios.
We utilize daily log returns to obtain an input cloud for filtration and, consequently, to compute the distance measures for clustering. At the same time, we employ simple returns to solve the portfolio optimization problem. 


For out-of-sample performance analysis, we follow a rolling window scheme.\footnote{We look at the problem of constructing a sparse optimal portfolio for a period $1,\ldots, T$, which corresponds to the recent past (the in-sample period) and where the path of asset returns is known. The ultimate aim is to create portfolios close to optimum performance in the $T + 1,\ldots, T + L$ period (the out-of-sample period). } To elaborate, we begin our first window with a training period of 126 trading days (six months) as the in-sample period to create the optimal portfolio and track its performance over a testing period of 21 trading days (one month) without any re-balancing during this period.{\footnote{We follow the standard of 21 trading days in a month (see, for instance, \cite{fastrich2014cardinality}).}} In the next window, we shift forward our training and testing period by one month to include more recent data and exclude the oldest data points such that a fixed size of the rolling window (seven months) is maintained. We then obtain an optimal portfolio on the new training period and track its performance over the new testing period. We repeat this procedure until the entire data sample is exhausted. Finally, we analyze the performance of the optimal portfolios over the testing period, thereby offering a comprehensive assessment of the effectiveness of the filtering strategy. 




To summarize, within a fixed rolling window, we obtain the time series of log returns for each constituent of the S\&P 500 index during the in-sample period. We then follow the following steps:
\begin{enumerate}
    \item[(i)] Convert each time series into a point cloud to calculate kernels $K_2$ and $K_4$. Divide the time series into sub-series and convert each into a point cloud to calculate kernels $K_1$ and $K_3$.
    
    \item[(ii)] Use these point clouds to generate corresponding persistence diagrams and persistence landscapes.
    
    \item[(iii)] Calculate pair-wise distances between all possible pairs of stock return time-series, as defined in Section \ref{newdistances}.
    
    \item[(iv)] Finally, obtain similarity matrices using the kernels $K_i,~i=1,2,3,4$. The calculations for kernels $K_5$ and $K_6$ are straightforward.
    
    \item[(v)] With each of the obtained similarity matrices, apply affinity propagation clustering, resulting in clusters.
    
    \item[(vi)] Select assets from the obtained clusters using the proposed strategies, Strategy 1 and Strategy 2. Solve the respective optimization problems using simple returns of the constituents to get optimal portfolios.
    
    \item[(vii)] Finally, track the performance of these portfolios over the out-of-sample period.
\end{enumerate}

\subsection{Strategy 1: Sparse Index Tracking}\label{sparse_index_tracking}

In this subsection, we assess the performance of our TDA-based data-driven sparse index tracking portfolios. We denote the optimal tracking portfolio corresponding to a similarity matrix $K_i$ by $C_i$, $i=1,2,\ldots,6$. We also create the following portfolios from the literature as benchmarks for comparative analysis.

\begin{itemize}
\item[1.] $\text{IT with cardinality constraint}$: 
    This portfolio corresponds to the optimization problem (\ref{tev1}) with an additional constraint imposed on the maximum number of assets that can be selected within a tracking portfolio. Specifically, we impose the following restriction that the number of holdings with non-zero weights $\leq K_{max}$ where we set $K_{max}$ equal to the number of assets selected in the portfolio $C_i,~i=1,2,\ldots,6$, to ensure a comparable number of asset for a fair comparison to the clustering based portfolios. Consequently, this yields six portfolios, denoted by $I_i,~i=1,2,\ldots,6$. Given that the problem (\ref{tev1}) with an additional constraint on the number of assets is a mixed-integer quadratic programming (MIQP) problem, it poses a challenging computational task. To tackle this computational bottleneck, we employ the Gurobi optimizer, which is recognized for its prowess in solving mixed-integer programming (MIP). For practical considerations, we follow \cite{shu2020high} and impose time limits of thirty minutes and one hour for solving each MIQP within every window. To distinguish, we denote the portfolios solved within a 30-minute time limit by $I_{i}^{(1)}$, while those solved with a one-hour time limit by $I_{i}^{(2)}$.

\item[2.] $\text{IT with similarity based selection}$: Considering that the clustering-based portfolios $C_i$ are constructed by selecting assets from the clusters identified using similarity matrices, an alternative and more straightforward approach could be to eliminate the need for clustering altogether and choose the assets with the maximum similarity to the index. To settle this debate, we introduce an additional set of six portfolios, denoted by $MS_i,i=1,2,\ldots,6$. Each $MS_i$ corresponds to the portfolio obtained by solving (\ref{tev1}) on top $M$ assets selected based on maximum similarity with the index, as measured by the similarity matrix $K_i$. We consider the maximum similarity criteria with $M=20$ as a benchmark.


\item[3.] $\text{IT with in-sample tracking error}$: In light of the objective function in (\ref{tev1}), i.e., the tracking error between the portfolio and the benchmark index, one might also consider the following simple approaches: (i) select $M$ assets that have the highest similarity to the index, as gauged by the in-sample tracking error between the stock return and the index return, (ii) follow the strategy one based on the similarity matrix, denoted by $K_7$, where similarity is measured by the negative of in-sample tracking error\footnote{$K_7$ is a negative squared distance function and is not a kernel (see Proposition 9.4 in \cite{paulsen2016introduction}) while the kernels $K_1,\ldots,K_6$ are actual kernel (positive definite) functions. Albeit for the affinity propagation clustering, positive definiteness is not required \cite{dueck2007non}.}: 
\begin{equation*}
 K_7(X,Y)=-||X-Y||_2^2.   
\end{equation*}
We denote this portfolio by $C_7$. 
\end{itemize}


Before analyzing the performance of the portfolios described above, we first examine the number of assets in the cluster containing the index resulting from APC based on the seven similarity matrices, \(K_i,~i=1,2,\ldots,7\), following our strategy 1. The mean and median number of assets in all windows are reported in Table \ref{clustering_2}. There are several observations. The number of clusters formed with TDA-based similarity matrices $K_1$ to $K_4$, as well as the correlation-based similarity matrices $K_5$ and $K_6$, is consistent over the sample period. A small number, on average, of assets in these clusters results in sparse portfolios as desired. For instance, the similarity matrix $K_2$ selects 40 assets (approximately 8\%) on average, while the correlation-based matrix $K_5$ (or $K_6$) selects a further smaller number of assets. In contrast, the similarity matrix $K_7$, on average, selects almost 57\% of assets, which is not desirable, and consequently, we drop the portfolio $C_7$ from our analysis.\footnote{In Appendix \ref{similarity_based_selection}, we present the study of the properties of the selected assets with respect to the benchmark index.}


\begin{table}[t!]
  \centering
  \caption{Mean and Median number of assets sharing cluster with the index}
    \begin{tabular}{l|rrrrrrr}
\hline
    Similarity matrix & $K_1$&$K_2$&$K_3$&$K_4$&$K_5$&$K_6$&$K_7$\\
\hline
Mean &7	&40&	6&	12&	18&	18&255 \\
Median&	6	&38&	5&	12&	18&	17&256 \\
 \\
    \end{tabular}%
  \label{clustering_2}%
\end{table}%

We also analyze the stability of cluster composition in time (across different windows) as well as among different similarity matrices, $K1$ to $K6$. To this end, we analyze the clustering structure over the period September 2020 to August 2022. We adopt a rolling window approach with a 6-month window size, shifting by one month, resulting in a total of 18 windows. We chose this recent period for its stable conditions and practical number of windows it offers, thus, making it easier to observe the variations in cluster composition and providing a reliable foundation for analysis. We start by calculating the Adjusted Rand Index (ARI) between the clustering composition in two consecutive windows for each similarity matrix.\footnote{Refer Appendix for more details on ARI.} We report the average value of ARI across the windows for each similarity matrix in Table \ref{ad_index}. It is evident from the Table that each similarity matrix except $K4$ generates an average ARI above 0.9638 suggesting that these similarity matrices generate relative stable clusters over consecutive windows. In contrast, $K4$, with an average ARI of 0.8301, shows slightly lower stability, indicating that its clustering assignments may vary more significantly across different windows. 
\begin{table}[ht!]
    \centering
    \caption{Average Adjusted Rand Index (ARI) for affinity propagation clustering using $K_1-K_6$ across consecutive windows}
    \begin{tabular}{lcccccc}
        \toprule
        \textbf{Similarity Matrix} & \textbf{$K1$} & \textbf{$K2$} & \textbf{$K3$} & \textbf{$K4$} & \textbf{$K5$} & \textbf{$K6$} \\
        \midrule
        \textbf{Average ARI} & 0.9875 & 0.9638 & 0.9779 & 0.8301 & 0.9858 & 0.9857 \\
        \bottomrule
    \end{tabular}
    \label{ad_index}
\end{table}

\begin{figure}[ht!]
\begin{center}
		\includegraphics[scale=.50]{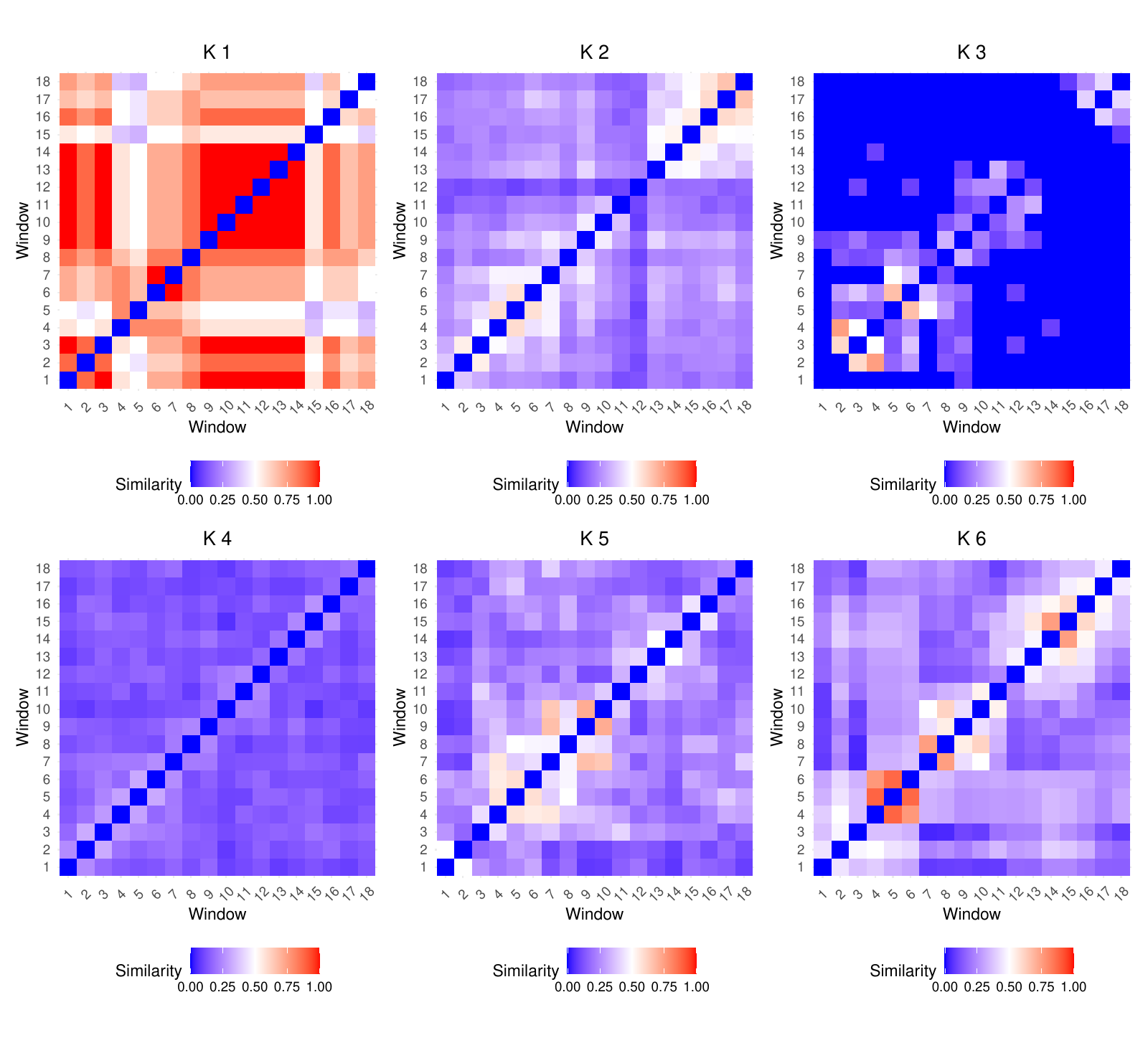}
     \caption{A comparison of Jaccard similarity between the clusters, from different windows, containing the index for six kernel matrices.}\label{fig:similarity_win}
\end{center}
\end{figure}

Next, we compare the composition of the cluster containing the benchmark index, that is, the cluster selected for the index tracking problem. To do this, we compute the Jaccard similarity between these clusters from different windows, with results displayed in Figure \ref{fig:similarity_win}.\footnote{Refer Appendix for more details on Jaccard Similarity.} Notably, we observe considerable variability in the similarity values across windows for each kernel, demonstrating how the cluster composition vary over time. There are several observations. First, the TDA based kernels, $K2$ to $K4$, show less similarity when compared to $K_1$ that produces relatively stable clusters. That is, the kernels $K_2$ to $K_4$ are able to capture the nuanced, non-linear relationships among assets, resulting in dynamic clusters that more accurately reflect market fluctuations in contrast to $K_4$ that is less responsive to market shifts. Second, the similarity matrix $K_2$ appears to strike a balance between stability and adaptability in the sense that the clusters in consecutive windows are neither the same nor completely changed, as evident by the first sub anti-diagonal (sub-diagonal above the main diagonal).\footnote{Investors seeking a combination of responsiveness and stability may find $K_2$ particularly advantageous. The stability of $K_2$ is further supported by its lower turnover values compared to other TDA-based portfolios, underscoring its suitability for maintaining a stable investment strategy. Refer results section for more details.} On the other hand, the kernels $K_3$ and $K_4$ shows significant change in the cluster composition, with exception by $K_3$ in the first half of the period. Third, the traditional similarity matrices, i.e., $K_5$ and $K_6$ also generate stable clusters, with $K_5$ demonstrating stability over a shorter period, while $K_6$ maintains stability in clusters of months (observe the small squares of different colors around the diagonal). Finally, observe the two prominent squares along the diagonal in the heat map of kernel $K_2$, the first one spanning the windows 1 to 12 and the second one covering the remaining six windows. These patterns corresponds to the bull (growth) and bear (fall) phase of the benchmark index (note that these clusters contain the benchmark) as illustrated in Figure \ref{fig:index}. This separation is less distinct in the other similarity matrices, making $K_2$ particularly effective at highlighting shifts in asset behavior in response to market shifts. To conclude, the ability of the similarity matrix $K_2$ to capture the regime shifts advocates for the strength of the TDA-based approach in identifying nuanced market patterns.

\begin{figure}[ht!]
\begin{center}
		\includegraphics[height=5.0cm, width=10.5cm]{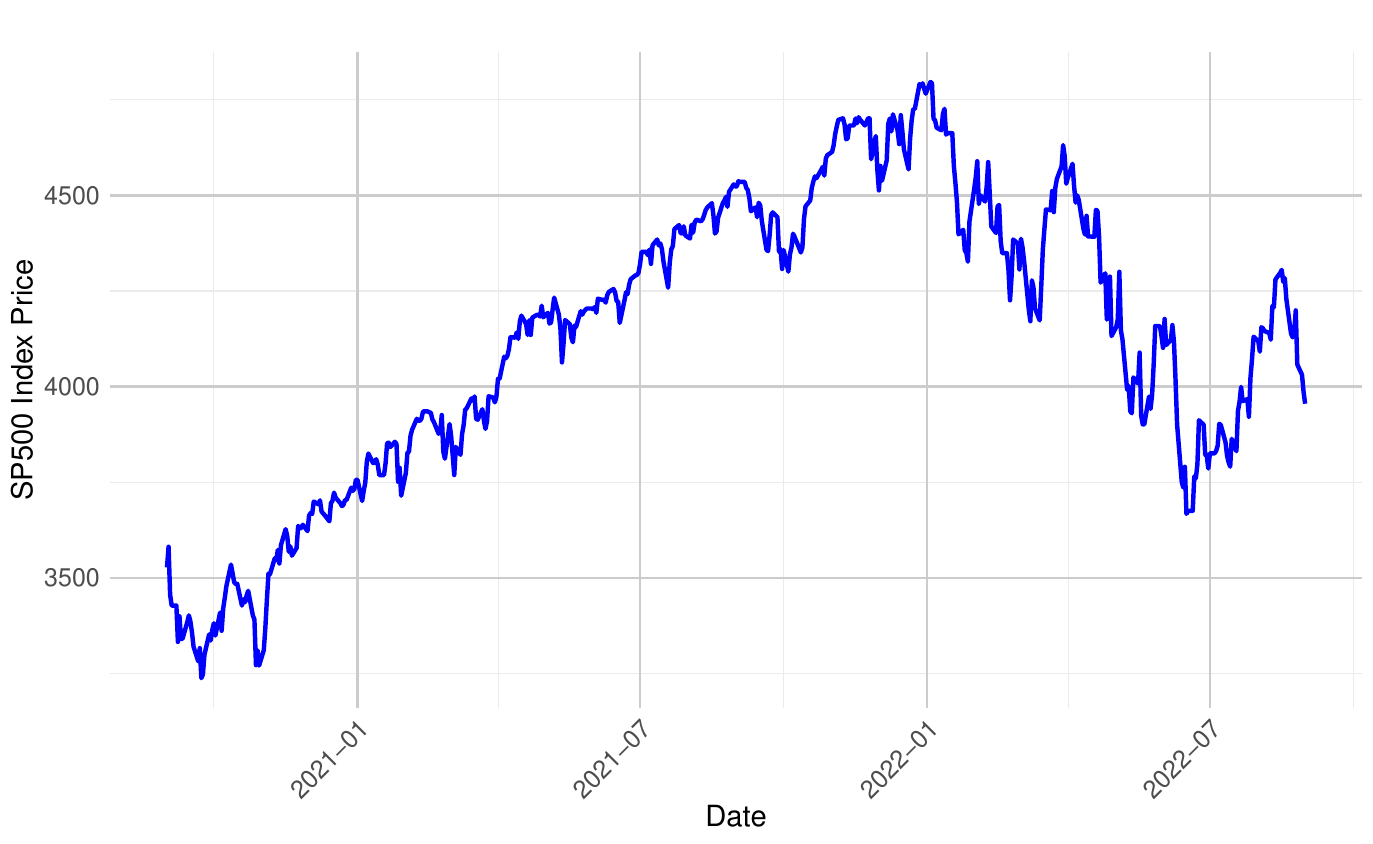}
    	
     \caption{Evolution of S\&P500 over the period, September 2020 to August 202.}\label{fig:index}
\end{center}
\end{figure}

Finally, we compare the composition of the cluster containing the benchmark index across the six similarity matrices, $K_1$ to $K_6$. To this end, we compute the Jaccard similarity between clusters from each pair of kernels for each window, then take the average similarity across all windows. The results are presented in Figure \ref{fig:kernels}. The relatively low similarity scores, ranging from 0 to 0.5, suggest that each similarity matrix captures distinct aspects of the data, generating clusters that differ significantly from one another, leading to diverse clusters. Notably, $K_5$ and $K_6$ show a moderate similarity compared to other pairs since they both try to capture the correlation among the assets. When comparing TDA-based similarity matrices to correlation-based matrices, $K_2$ shows the highest similarity to both $K_5$ and $K_6$. This suggests that TDA-based approaches could provide a valuable alternative to traditional methods, capturing not only the correlation but potentially provide deeper and robust insights into the asset relationships, particularly when dealing with non-linear dependencies.

\begin{figure}[ht!]
\begin{center}
		\includegraphics[scale=.46]{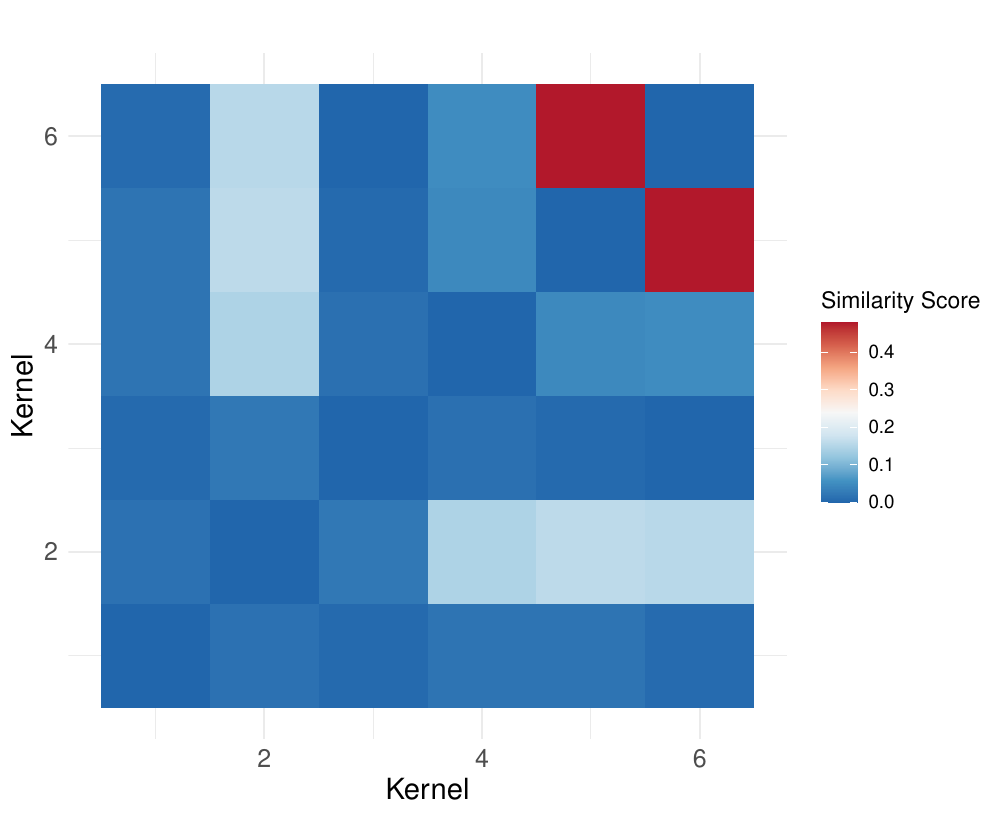}
    	
     \caption{Average similarity (Jaccard) across the cluster composition (containing the benchmark index) across the
six similarity matrices, $K_1$ to $K_6$. }\label{fig:kernels}
\end{center}
\end{figure}

\begin{table*}[t!] 
  \centering
  \caption{The out-of-sample performance of clustering based portfolios $C_1-C_6$, maximum similarity based portfolios $MS_1-MS_6$ and portfolios from TEV model with cardinality constraint $I_{1t_1}-I_{6t_1}$ and $I_{1t_2}-I_{6t_2}$.}
    \begin{tabular}{l|rrrrrr}
\hline
     & \multicolumn{1}{l}{$C_1$} & \multicolumn{1}{l}{$C_2$} & \multicolumn{1}{l}{$C_3$} & \multicolumn{1}{l}{$C_4$} & \multicolumn{1}{l}{$C_5$} & \multicolumn{1}{l}{$C_6$} \\
\hline
    TE    & 2.62E-05 & \textit{5.14E-06} & 2.79E-05 & \textbf{5.00E-06} & 1.21E-05 & 1.08E-05 \\
    EMR   & \textbf{1.44E-04} & \textit{1.30E-04} & -1.05E-05 & 3.40E-05 & 9.93E-05 & 9.03E-05 \\
    COR   & 8.36E-01 & \textit{9.68E-01} & 8.32E-01 & \textbf{9.70E-01} & 9.39E-01 & 9.44E-01 \\
    INFO  & 2.81E-02 & \textbf{5.71E-02} & -1.98E-03 & 1.52E-02 & \textit{2.86E-02} & 2.75E-02 \\
    TR & 1.47   & \textit{1.10  } & 1.70   & 1.58   & 1.18   & \textbf{1.08  } \\
\hline
     & \multicolumn{1}{l}{$MS_1$} & 
     \multicolumn{1}{l}{$MS_2$} &
     \multicolumn{1}{l}{$MS_3$} &
     \multicolumn{1}{l}{$MS_4$} &
     \multicolumn{1}{l}{$MS_5$} &\multicolumn{1}{l}{$MS_6$} \\
\hline
    TE   & 1.01E-05 &\textit{7.61E-06}& 1.18E-05 &\textbf{6.85E-06}& 8.79E-06& 8.76E-06  \\
    EMR    & 3.71E-05 & 5.74E-05 & -3.90E-05 & \textbf{6.62E-05} & 6.00E-05 & 5.91E-05 \\
    COR   & 9.04E-01 & 9.36E-01 & 8.85E-01 & \textbf{9.37E-01} & 9.37E-01 & 9.41E-01 \\
    INFO  & 9.55E-03 & 1.77E-02 & -9.07E-03 & \textbf{2.02E-02} & 1.63E-02 & 1.71E-02 \\
    TR & 1.02   & \textbf{7.74E-01} & 1.34   & 1.47   & 9.54E-01 & 8.50E-01 \\
\hline
     & \multicolumn{1}{l}{$I_{1t_1}$} & \multicolumn{1}{l}{$I_{2t_1}$} & \multicolumn{1}{l}{$I_{3t_1}$} & \multicolumn{1}{l}{$I_{4t_1}$} & \multicolumn{1}{l}{$I_{5t_1}$} & \multicolumn{1}{l}{$I_{6t_1}$}  \\
\hline
TE    & 5.70E-05 & 3.14E-05 & 3.20E-05 & 5.41E-05 & \textbf{2.55E-05} & 4.67E-05 \\
    EMR   & -5.81E-05 & 1.90E-04 & -1.54E-04 & \textbf{2.25E-04} & -3.14E-05 & -1.67E-05 \\
    
    COR   & 7.57E-01 & 8.52E-01 & 8.43E-01 & 7.66E-01 & \textbf{8.75E-01} & 8.01E-01 \\
    INFO  & -7.69E-03 & \textbf{3.39E-02} & -2.72E-02 & 3.06E-02 & -6.22E-03 & -2.45E-03 \\
    TR & 1.72   & \textbf{1.61  } & 1.65   & 1.64   & 1.68   & 1.62   \\
\hline
  &  \multicolumn{1}{l}{$I_{1t_2}$} & \multicolumn{1}{l}{$I_{2t_2}$} & \multicolumn{1}{l}{$I_{3t_2}$} & \multicolumn{1}{l}{$I_{4t_2}$} & \multicolumn{1}{l}{$I_{5t_2}$} & \multicolumn{1}{l}{$I_{6t_2}$}  \\
\hline
    TE    & 1.34E-05 & \textbf{2.45E-06} & 1.28E-05 & 8.00E-06 & 6.57E-06 & 5.13E-06 \\
    EMR   & \textbf{1.36E-04} & 3.11E-05 & -3.70E-06 & 9.91E-06 & 7.48E-05 & 4.76E-05 \\
    
    COR   & 9.26E-01 & \textbf{9.85E-01} & 9.32E-01 & 9.54E-01 & 9.62E-01 & 9.71E-01 \\
    INFO  & \textbf{3.73E-02} & 1.99E-02 & -1.03E-03 & 3.50E-03 & 2.92E-02 & 2.10E-02 \\
    TR & 1.63   & \textbf{1.58  } & 1.66   & 1.71   & 1.68   & 1.66   \\
    \end{tabular}%
  \label{performance_1}%
\end{table*}%


Table \ref{performance_1} presents the out-of-sample performance of all portfolios in terms of five performance measures, i.e., correlation with the index (COR), tracking error (TE), excess mean return (EMR), information ratio (INFO), and the turnover ratio (TR). The best and second-best values among models are highlighted in bold and italics, respectively. 
The average computation time (from data to portfolio weights) per rolling window for portfolios $C_1$ to $C_4$ are 10.36, 15.09, 11.46, and 16.15 minutes, respectively. In contrast, the computational time for portfolios $C_5$ and $C_6$, based on correlation coefficients, are 1.5 and 2.47 minutes. The time for $I_{t_1}$ and $I_{t_2}$ are respectively 30 minutes and 60 minutes by definition of these portfolios. The computational time for $C_1$ to $C_4$ is relatively larger since these portfolios involves calculating the Wasserstein distance between persistence diagrams. While TDA based portfolios take more time, these durations are manageable given the fact the obtained portfolios are rebalanced on monthly basis, thus, making the computational load relatively insignificant. Given that it is challenging to compare all the models simultaneously, we first identify the best-performing portfolio among those constructed based on clustering using TDA-based distance measures. Subsequently, we comprehensively compare these selected portfolios to those based on maximum similarity criteria and cardinality constraints.

First, it is evident that all but $C_3$ portfolios generate positive EMR. Although an ideal tracking portfolio should have zero excess mean return, it is rarely observed in practice. We usually observe both positive and negative values with a positive (negative) EMR indicating that the tracking portfolio outperforms (under-performs) the benchmark index. A positive return over the benchmark is desirable as it can certainly offset the negative impact of transaction costs. Second, in terms of the tracking performance, the portfolios $C_4$ and $C_2$ outperform the other TDA-based clustering portfolios in terms of out-of-sample tracking error and result in the highest correlation of 96.8\% and 97\%, respectively, with the index.

Furthermore, the tracking error from $C_5$ and $C_6$ is relatively high compared to the tracking error from $C_2$ and $C_4$, demonstrating the proposed kernels' efficacy in filtering out more identical assets than the conventional kernels. Third, $C_2$ excels $C_4$ as it attains the highest information ratio and maintains a trade-off between the return over and above the benchmark index and the tracking error. Last but not least, we compare the turnover values that reflect the potential transaction costs associated with portfolio re-balancing, with higher values showing more re-balancing, hence higher transaction costs, and vice versa. The turnover values show that the investor needs to re-balance fewer positions for $C_6$ and $C_2$ compared to the re-balancing for $C_4$ portfolios from one investment period to the other. Finally, we conclude that both $C_2$ and $C_4$ outperform the other clustering-based portfolios, with $C_2$ outperforming $C_4$ in terms of turnover ratio and information ratio. 

Next, we observe that both the portfolios $C_2$ and $C_4$ outperform the portfolios based on maximum similarity selection criteria with $M=10$ in terms of the tracking error and correlation with the index. Further, $C_2$ also beats the maximum similarity portfolios in terms of EMR, while $C_4$ fails to do so. One can attribute the smaller values of TR for maximum similarity portfolio compared to $C_2$ and $C_4$ to the fact the $MS$ portfolios consist of only ten assets and hence the less re-balancing.

Comparing the performance of the portfolios $I_{1t_1}-I_{6t_1}$ to that of the portfolios $I_{1t_2}-I_{6t_2}$, we find that the latter set of portfolios outperforms the former in terms of TE, EMR, and correlation with the index. In contrast, for other performance measures, INFO and TR, there is no clear domination. Further, the portfolios $I_{1t_1}-I_{6t_1}$ perform poor than the TDA-based sparse portfolios. When we compare the clustering-based portfolios to portfolios represented as $I_{1t_2}-I_{6t_2}$, we observe that the tracking error from $C_2$ is greater than that of $I_{2t_2}$, but the turnover is relatively low. This higher tracking error in the case of $C_2$ is quite intuitive. We minimize tracking error by utilizing all available assets when considering $I_{2t_2}$. In contrast, with $C_2$, we first filter out a subset of assets and then minimize the tracking error on this reduced set, resulting in a sub-optimal solution. However, it's important to note that the former approach incurs higher transaction costs because, in each evaluation window, the optimization problem is solved anew using all assets, leading to frequent portfolio re-balancing. 

On the other hand, for $C_2$, the initial filtering step imparts stability to the problem, resulting in less frequent portfolio re-balancing and, consequently, lower transaction costs. In conclusion, no consensus exists for the best among $C_2$, $C_4$, and $I_{1t_2}-I_{6t_2}$. There is a trade-off between the tracking error and the transaction costs. Nonetheless, clustering-based portfolios require less computational time when compared to cardinality-constrained portfolios. To find the best portfolio, we further analyze their performance. That is, given the superior performance of the portfolios $C_2$, $C_4$ and $I_{1t_2}-I_{6t_2}$, we now compute the mean ratio of the in-sample tracking error and out-of-sample tracking error \footnote{We calculate the ratio of in-sample tracking error and out-of-sample tracking error in each window and then take average to find the mean \cite{beas2009}}. We also compare it to the mean ratio of the portfolio $TE_{all}$, which corresponds to a complete replication index tracking portfolio, i.e., a tracking portfolio constructed using all the assets constituting the index without penalty terms. A comparison of the mean ratios allows us to investigate the severity of the estimation error. Table \ref{con} shows the average value of in-sample and out-of-sample tracking error across the windows and the average value of the corresponding ratios. A ratio of less than one across all portfolios indicates that, on average, there is a performance decline when transitioning from the in-sample period to the out-of-sample period. It is not surprising since the portfolio is optimized over the in-sample period, thus generating a smaller tracking error, while the same portfolio composition, when tracked over the out-of-sample period, the data that is unseen before. Nevertheless, it is remarkable that performance degradation is not excessively high, particularly for portfolios based on clustering techniques. This reaffirms the robustness of clustering-based portfolios in both the in-sample and out-of-sample periods.

\begin{table*}[t!]
  \centering
  \caption{In-sample and out-of-sample robustness check. In-TE is the average in-sample value of the tracking error across the windows, Out-TE is the average out-of-sample value of the tracking error across the windows and Mean ratio is the mean (taken across the windows) value of ratio of in-sample and out-of-sample tracking error. $C_2$ and $C_4$ are clustering based portfolios, $I_{1t_2}-I_{6t_2}$ are portfolios from cardinality constraint model and $TE_{all}$ is the portfolio build using all the assets.}
   \scalebox{0.8}{
    \begin{tabular}{l|rr|rrrrrr|r}
          & \multicolumn{2}{l|}{Clustering} &&     &        \multicolumn{3}{l}{Cardinality constraint}             &       &        \multicolumn{1}{|l}{All} \\
          \hline
          & \multicolumn{1}{l}{$C_2$} & \multicolumn{1}{l|}{$C_4$} & \multicolumn{1}{l}{$I_{1t_2}$} & \multicolumn{1}{l}{$I_{2t_2}$} & \multicolumn{1}{l}{$I_{3t_2}$} & \multicolumn{1}{l}{$I_{4t_2}$} & \multicolumn{1}{l}{$I_{5t_2}$} & \multicolumn{1}{l|}{$I_{6t_2}$} & \multicolumn{1}{l}{$TE_{all}$} \\
\hline
    In TE & 2.4E-06 & 2.2E-06 & 5.1E-06 & 1.7E-07 & 4.8E-06 & 2.1E-06 & 2.4E-06 & 1.0E-06 & 5.4E-15 \\
    Out TE & 5.1E-06 & 5.0E-06 & 1.3E-05 & 2.4E-06 & 1.3E-05 & 8.0E-06 & 6.6E-06 & 5.1E-06 & 9.0E-07 \\
Mean ratio& \textbf{5.5E-01} & \textit{5.2E-01} & 4.4E-01 & 7.4E-02 & 4.7E-01 & 3.0E-01 & 2.3E-01 & 2.3E-01 & 7.1E-09 \\
    \end{tabular}}%
 \label{con}
\end{table*}%

\begin{rem}
To further validate the exceptional performance of our top-performing TDA-based sparse portfolio, $C_2$, we compare it to the tracking portfolios constructed using the co-integration technique, a methodology known for its effectiveness in index-tracking problems \cite{alexander2005indexing}, \cite{sant2017indexQREF}. We denote this portfolio by $CO$. The idea underlying the co-integration index tracking portfolio revolves around determining the optimal weights for a given set of assets $S_1,\ldots,S_n$ that achieve the highest level of co-integration with the benchmark. Nevertheless, 
selecting the parameter $n$ is vital and demands careful consideration. For instance, to obtain a 
stable co-integrating vector of at least $n$ stocks out of $N$ stocks in the index, a mechanical search would entail examining an extensive array of portfolios - specifically, $N!/(N-n)!n!$ portfolios - as potential candidates for backtesting. However, in our paper, we adopt the approach presented in \cite{sant2017indexQREF} for constructing a co-integration-based index-tracking portfolio. More specifically, we start by randomly selecting 40 stocks, a choice motivated by our intention to compare it to the $C_2$ portfolio, which, on average, consists of 40 stocks. We run a regression analysis (for a comprehensive understanding of the co-integration model, please refer to \cite{sant2017indexQREF}) on the in-sample data and select this portfolio as a candidate if the residuals from the regression are stationary, as indicated by the Augmented Dickey-Fuller (ADF) test. The stationary residuals would ascertain the presence of co-integration between the index and asset price series. Otherwise, we discard the portfolio under consideration. We repeat this process for 100,000 regressions conducted sequentially for each portfolio. Each regression employs a unique combination of 40 assets. The final portfolio chosen is determined by the regression that yields the smallest sum of squared residuals. Our findings suggest that the $C_2$ portfolio outperforms the $CO$ portfolio in terms of all the performance metrics as reported in Table \ref{coint}. For instance, the $C_2$ portfolio achieves a significantly lower tracking error of 5.14E-06 in contrast to 1.57E-05 achieved by the $CO$ portfolio. Similarly, the correlation of 0.968 with the index for the $C_2$ portfolio is significantly higher than 0.924 for the $CO$ portfolio. These findings suggest that TDA-based sparse index tracking portfolios perform better than traditional portfolios based on correlation and co-integration with the index.
\end{rem}

\begin{table}[ht!]
  \centering
  \caption{The out-of-sample performance of $C_2$ portfolio constructed using Strategy 1 with kernel 2 and the $CO$ portfolio obtained from co-integration based index tracking model.}
    \begin{tabular}{l|rr}
         \hline
          &   \multicolumn{1}{l}{$C_2$}    & \multicolumn{1}{l}{$CO$} \\
 \hline  
    TE    & \textbf{5.14E-06} & 1.57E-05 \\
     EMR   & \textbf{1.30E-04} & 5.06E-05 \\
    COR   & \textbf{9.68E-01} & 9.24E-01 \\
    INFO  & \textbf{5.71E-02} & 1.28E-02 \\
    TR & \textbf{1.10E+00} & 2.38E+00 \\
    \end{tabular}%
  \label{coint}%
\end{table}%


Overall, we observed an improvement in the index tracking problem using the TDA-based kernels with strategy 1. The key concern, though, is the robustness of the strategy. The market fluctuations are not deterministic and cannot be captured. So, we analyze the performance of TDA-based kernels during the crisis period. In particular, we look at the recent Covid-19 pandemic.

\textbf{Covid-19 robustness check:} As a robustness check, Table \ref{covid_indextracking} shows the out-of-sample performance of tracking portfolios for data from the recent COVID-19 pandemic, August 2019 to August 2020. Again, we use the daily prices of the S\&P500 index and its 502 constituents as available in this period. We adopt the rolling window strategy with 21 days (approx one month) in-sample period and five days (approx one week) out-of-sample period. Among TDA-based portfolios, we observe that $C_2$ and $C_4$ still have the lowest and statistically significant tracking errors. Also, $C_1$ attains minimum turnover while $C_2$ attains second minimum turnover compared to other portfolios, an observation that aligns with observations in Figure \ref{fig:similarity_win} where we observed that $C_1$ generates the most similar clusters.

In contrast to the stable period, the best performing TDA-based portfolios, i.e., $C_2$ and $C_4$ outperforms the maximum similarity-based portfolio $MS_2$ outperforms in terms of tracking error and EMR, but falls short in terms of correlation and turnover. Specifically, $C_2$ and $C_4$ achieves tracking errors of 2.43E-03, and 2.56E-03, respectively, which is better than 4.16E-03 recorded by $MS_2$. However, $MS_2$ shows a favorable turnover ratio of 0.77, compared to 1.39 and 1.98 for $C_2$ and $C_4$ respectively. This is likely because $MS_2$ maintains a fixed portfolio size of 20 assets, thereby capping the maximum number of position changes at 20. The comparison of TDA-based portfolios against cardinality-based portfolios shows that those based on the similarity matrix $K_2$ outperforms other cardinality based portfolios. Additionally, TDA based portfolio $C_2$ outperforms the best cardinality based portfolios, $I_{2t_1}$ and $I_{2t_2}$ by achieving a lower tracking error and turnover, alongside higher correlation. Therefore, considering the trade-off between the performance and the computational time, clustering-based portfolios emerge as the superiod alternative.


\begin{table*}[t!] 
  \centering
  \caption{The out-of-sample performance of clustering based portfolios $C_1-C_6$, maximum similarity based portfolios $MS_1-MS_6$ and portfolios from TEV model with cardinality constraint $I_{1t_1}-I_{6t_1}$ and $I_{1t_2}-I_{6t_2}$ on the Covid-19 pandemic period.}
    \begin{tabular}{l|rrrrrr}
    \hline
     & \multicolumn{1}{l}{$C_1$} & \multicolumn{1}{l}{$C_2$} & \multicolumn{1}{l}{$C_3$} & \multicolumn{1}{l}{$C_4$} & \multicolumn{1}{l}{$C_5$} & \multicolumn{1}{l}{$C_6$} \\
     \hline
TE           & 7.55E-03 & \textbf{2.43E-03} & 7.65E-03 & \textit{2.56E-03} & 5.69E-03 & 6.91E-03 \\
EMR          & -8.93E-04 & 8.22E-05 & \textbf{4.33E-04} & -6.46E-05 & -1.38E-04 & \textit{1.35E-04} \\
COR  & 9.42E-01 & \textbf{9.96E-01} & 9.36E-01 & \textit{9.91E-01} & 9.77E-01 & 9.67E-01 \\
INFO          & -1.18E-01 & \textit{3.38E-02 }& \textbf{5.66E-02} & -2.52E-02 & -2.43E-02 & 1.95E-02 \\

  TR &\textbf{ 0.88 }  & \textit{1.39 }& 1.98   & 1.59   & 1.74   & 1.71   \\
    \hline
  & \multicolumn{1}{l}{$MS_1$} & 
     \multicolumn{1}{l}{$MS_2$} &
     \multicolumn{1}{l}{$MS_3$} &
     \multicolumn{1}{l}{$MS_4$} &
     \multicolumn{1}{l}{$MS_5$} &\multicolumn{1}{l}{$MS_6$}  \\
     \hline
TE           & 6.47E-03 & \textbf{4.16E-03} & 6.78E-03 & \textit{4.65E-03} & 5.00E-03 & 4.9E-03 \\
EMR          & -5.82E-04 & \textbf{5.42E-05} & 8.05E-06 & -1.23E-03 & \textit{5.15E-05} & 1.35E-05 \\
COR  & 9.79E-01 & \textbf{9.95E-01} & 9.71E-01 & \textit{9.91E-01} & 9.83E-01 & 9.88E-01 \\
INFO     &     -8.97E-02 &\textbf{ 1.30E-02 }& 1.19E-03 & -2.65E-01 & \textit{1.03E-02} & 2.76E-03 \\

   TR & 1.16   & \textbf{0.77  } & 1.56   & 1.67   &\textit{ 0.86 }  & 1.02   \\
    \hline
      &  \multicolumn{1}{l}{$I_{1t_1}$} & \multicolumn{1}{l}{$I_{2t_1}$} & \multicolumn{1}{l}{$I_{3t_1}$} & \multicolumn{1}{l}{$I_{4t_1}$} & \multicolumn{1}{l}{$I_{5t_1}$} & \multicolumn{1}{l}{$I_{6t_1}$}  \\
\hline    
  TE    & 7.14E-03 & \textbf{3.23E-03} & 6.15E-03 & \textit{3.24E-03} & 1.12E-02 & 5.57E-03 \\
EMR   & -4.71E-04 & \textbf{1.67E-04} & -9.76E-05 &\textit{ -9.55E-05} & -2.91E-04 & -1.44E-04 \\
COR   & 9.68E-01 & \textbf{9.93E-01} & 9.76E-01 & \textit{9.93E-01} & 9.15E-01 & 9.81E-01 \\
INFO    & -6.60E-02 & \textbf{5.16E-02} & \textit{-1.59E-02} & -2.95E-02 & -2.60E-02 & -2.58E-02 \\
  TR    & 1.63  & 1.52   & \textbf{1.45  } & \textit{1.51  } & 1.87   & 1.79   \\
  \hline
      &  \multicolumn{1}{l}{$I_{1t_2}$} & \multicolumn{1}{l}{$I_{2t_2}$} & \multicolumn{1}{l}{$I_{3t_2}$} & \multicolumn{1}{l}{$I_{4t_2}$} & \multicolumn{1}{l}{$I_{5t_2}$} & \multicolumn{1}{l}{$I_{6t_2}$}  \\
\hline
  TE    & 6.84E-03 & \textbf{3.22E-03} & 6.05E-03 & \textit{3.24E-03} & 6.79E-03 & 5.64E-03 \\
EMR   & {-5.98E-04} & \textbf{1.58E-04} & -2.62E-04 & -9.55E-05 & -2.85E-04 & \textit{3.98E-06 }\\
COR   & 9.72E-01 & \textbf{9.93E-01} & 9.77E-01 & \textit{9.93E-01} & 9.68E-01 & 9.80E-01 \\
INFO    & -8.75E-02 & \textbf{4.90E-02} & -4.33E-02 & -2.95E-02 & -4.20E-02 & \textit{7.05E-04} \\

   TR &1.63 	&1.50  &\textbf{	1.41 } &	\textit{1.50  }	&1.68  &	1.77 
\\

    \end{tabular}%
 \label{covid_indextracking}
\end{table*}%

\textbf{Case Study (Publicly available data for Index tracking model):} For completeness, we also test the performance of our TDA-based sparse index tracking portfolios on the benchmark data sets from \cite{beasley2003evolutionary,beas2009} for index tracking problems. These data sets are drawn from several major world equity indices, and thus, analyzing the performance of these data sets gives additional evidence on the adequacy and effectiveness of the TDA-based portfolios. 

\begin{table*}[t!]
		\centering
		\caption{Indices used for empirical analysis\color{blue}}
		\scalebox{.8}{
			\begin{tabular}{l|lccc}
				\hline
				S. No. & Index & Country & Total Number &  Constituents  \\
    				 &  &  & of constituents &  considered \\
				\hline
			1. & FTSE 100 & UK & 100 & 89\\
			2. & Nikkei 225 & Japan & 225 & 225\\
			3. &  DAX 100& Germany& 100 & 85\\
			4. & Hang-Seng & Hong Kong & 50 & 31\\
            5. & S\&P 100 & USA & 100 & 98\\
			6. & S\&P 500 & USA  & 500 & 457\\
			7. &  Russell 2000& USA& 2000 & 1318\\
			8. & Russell 3000& USA&  3000 & 2151\\
	\end{tabular}
		}
		\label{dataset}%
\end{table*}
Our data sets belong to the OR-Library and are publicly available at \url{http://people.brunel.ac.uk/~mastjjb/jeb/info.html}. There are currently eight benchmark instances. Each instance corresponds to the securities composing a specific market index. Specifically, the following market indices are considered: Hang Seng (Hong Kong), DAX100 (Germany), FTSE100 (United Kingdom), S\&P100 (USA), Nikkei225 (Japan), S\&P500 (USA),
Russell2000 (USA) and Russell3000 (USA). The data set comprises 291 weekly closing prices for each of the indices spanning March 1992 to September 1997. Table \ref{dataset} gives the list of indices, their actual number of constituents, and the number of constituents considered for this study.

We follow \cite{beasley2003evolutionary} and consider the first 145 weekly prices as in-sample observations and the following 146 weeks as out-of-sample periods to evaluate the performance of the models. We do not require a rolling window approach here, as there is only one training and testing period. Consequently, instead of reporting transaction costs, we find the Hinderfindahl index, a diversity index calculated as $\text{HHI}=\sum_{i=1}^{n}w_i^2.$ In the case of a na\"{\i}ve portfolio (which offers maximum diversification), the minimum value of HHI is obtained and is equal to $\frac{1}{n}$. Smaller values of HHI are desirable.

	\begin{table*}[t!]
  \caption{Out-of-sample performance of strategy 1 on data set from Beasley OR Library. Here $M$ denotes the number of assets filtered as clustering result, i.e., the number of assets sharing the same cluster with the index }\resizebox{.52\textwidth}{!}{\begin{minipage}{\textwidth}
    \begin{tabular}{l|rrrrrrr|l|rrrrrrr}
\hline
    T1    & \multicolumn{1}{l}{$C_1$} & \multicolumn{1}{l}{$C_2$} & \multicolumn{1}{l}{$C_3$} & \multicolumn{1}{l}{$C_4$} & \multicolumn{1}{l}{$C_5$} & \multicolumn{1}{l}{$C_6$} & \multicolumn{1}{l|}{$TE_{all}$} & T5 &  \multicolumn{1}{l}{$C_1$} & \multicolumn{1}{l}{$C_2$} & \multicolumn{1}{l}{$C_3$} & \multicolumn{1}{l}{$C_4$} & \multicolumn{1}{l}{$C_5$} & \multicolumn{1}{l}{$C_6$} & \multicolumn{1}{l}{$TE_{all}$} \\
    \hline
    EMR   & -4.61E-05 & 9.65E-05 & -4.85E-04 & -3.76E-05 & 2.59E-04 & 9.16E-04 & 9.96E-05 & EMR   & -1.32E-03 & -7.80E-04 & -3.65E-05 & -1.04E-03 & -9.76E-04 & -1.24E-03 & -3.68E-04 \\
    TE    & \textbf{1.27E-05} & \textit{1.46E-05} & 1.84E-04 & 9.82E-05 & 1.96E-05 & 4.83E-05 & 7.30E-06 & TE    & \textit{5.85E-05} & \textbf{4.58E-05} & 1.53E-04 & 8.80E-05 & 6.34E-05 & 6.91E-05 & 2.96E-05 \\
    COR   & \textbf{9.92E-01} & \textit{9.91E-01} & 9.27E-01 & 9.36E-01 & 9.89E-01 & 9.79E-01 & 9.95E-01 & COR   & 9.66E-01 & \textit{9.70E-01} & 8.94E-01 & 9.43E-01 & \textbf{9.71E-01} & 9.61E-01 & 9.84E-01 \\
    HHI   & \textbf{8.81E-02} & \textit{8.82E-02} & 1.96E-01 & 1.56E-01 & 9.57E-02 & 1.21E-01 & 7.05E-02 & HHI   & \textbf{2.90E-02} & \textit{3.85E-02} & 2.60E-01 & 9.01E-02 & 4.92E-02 & 4.80E-02 & 1.37E-02 \\
    M     & 19    & 18    & 6     & 8     & 18    & 15    & 31    & M     & 62    & 54    & 5     & 15    & 57    & 46    & 225 \\
    \hline
    T2    &       &       &       &       &       &       &       & T6    &       &       &       &       &       &       &  \\
    \hline
    EMR   & 3.32E-04 & 1.39E-04 & -2.07E-04 & 7.67E-04 & 3.43E-04 & 1.78E-04 & -2.66E-05 & EMR   & 3.08E-03 & 2.15E-03 & 3.55E-03 & 2.34E-03 & 2.32E-03 & 1.97E-03 & 2.22E-03 \\
    TE    & 1.12E-04 & \textbf{7.96E-05} & \textit{1.79E-04} & 1.14E-04 & 9.76E-05 & 8.77E-05 & 5.76E-05 & TE    & 7.64E-04 & \textbf{2.85E-04} & 7.30E-04 & \textit{3.24E-04} & 5.78E-04 & 4.33E-04 & 1.03E-04 \\
    COR   & 8.61E-01 & \textbf{9.06E-01} & 8.05E-01 & 8.67E-01 & 8.86E-01 & \textit{8.95E-01} & 9.30E-01 & COR   & 6.37E-01 & \textit{7.89E-01} & 4.86E-01 & 7.71E-01 & 7.50E-01 & \textbf{8.01E-01} & 9.28E-01 \\
    HHI   & \textit{1.20E-01} & \textbf{1.12E-01} & 3.42E-01 & 1.58E-01 & 1.55E-01 & 1.26E-01 & 3.87E-02 & HHI   & 3.11E-01 & \textbf{5.67E-02} & 2.51E-01 & \textit{9.50E-02} & 2.65E-01 & 2.14E-01 & 1.80E-02 \\
    M     & 27    & 12    & 4     & 8     & 8     & 10    & 85    & M     & 4     & 40    & 7     & 24    & 6     & 7     & 457 \\
\hline
    T3    &       &       &       &       &       &       &       & T7    &       &       &       &       &       &       &  \\
    \hline
    EMR   & -8.44E-04 & -9.27E-04 & 2.47E-04 & -9.36E-04 & -1.23E-03 & -1.58E-03 & -1.72E-04 & EMR   & 2.32E-03 & 2.24E-03 & \textbf{3.02E-03} & 2.50E-04 & 6.02E-04 & -3.54E-04 & 1.79E-03 \\
    TE    & 7.94E-05 & \textbf{2.95E-05} & 8.27E-05 & 5.12E-05 & 5.25E-05 & \textit{4.16E-05} & 7.83E-06 & TE    & 5.65E-04 & \textbf{3.96E-04} & 7.80E-04 & 6.66E-04 & 5.39E-04 & 6.09E-04 & 2.14E-04 \\
    COR   & 8.38E-01 & \textbf{9.36E-01} & 8.24E-01 & 8.91E-01 & 8.84E-01 & \textit{9.11E-01} & 9.83E-01 & COR   & 7.07E-01 & \textbf{7.74E-01} & 5.20E-01 & 5.65E-01 & 7.18E-01 & 6.89E-01 & 8.93E-01 \\
    HHI   & 1.05E-01 & \textbf{4.36E-02} & 2.31E-01 & 9.42E-02 & \textit{5.85E-02} & 6.43E-02 & 2.26E-02 & HHI   & 1.09E-01 & \textbf{3.54E-02} & 1.53E-01 & 4.88E-02 & 6.48E-02 & 9.56E-02 & 1.33E-02 \\
    M     & 11    & 35    & 5     & 14    & 25    & 24    & 89    & M     & 14    & 63    & 8     & 47    & 40    & 20    & 1318 \\
    \hline
    T4    &       &       &       &       &       &       &       & T8    &       &       &       &       &       &       &  \\
    \hline
    EMR   & -4.47E-04 & -2.75E-04 & -9.72E-04 & -5.62E-04 & -4.80E-04 & -7.23E-04 & -1.26E-04 & EMR   & 1.43E-03 & 1.71E-03 & 3.69E-03 & 1.76E-03 & 1.67E-03 & 2.23E-03 & 2.52E-03 \\
    TE    & 3.63E-04 & \textbf{3.86E-05} & 1.65E-04 & \textit{5.31E-05} & 6.69E-05 & 7.03E-05 & 9.35E-06 & TE    & 6.49E-04 & \textbf{2.33E-04} & 7.31E-04 & 3.18E-04 & 5.93E-04 & 5.26E-04 & 1.14E-04 \\
    COR   & 6.30E-01 & \textbf{9.35E-01} & 7.44E-01 & \textit{9.16E-01} & 8.93E-01 & 8.95E-01 & 9.85E-01 & COR   & 6.09E-01 & \textbf{8.32E-01} & 3.21E-01 & 7.70E-01 & 7.17E-01 & 7.38E-01 & 9.26E-01 \\
    HHI   & 3.69E-01 & \textbf{7.24E-02} & 2.21E-01 & \textit{9.87E-02} & 1.06E-01 & 1.19E-01 & 2.70E-02 & HHI   & 1.49E-01 & \textbf{4.40E-02} & 2.09E-01 & 5.49E-02 & 2.29E-01 & 1.62E-01 & 1.63E-02 \\
    M     & 5     & 24    & 5     & 14    & 15    & 13    & 98    & M     & 8     & 62    & 8     & 50    & 8     & 13    & 2151 \\
    \end{tabular}%
    \label{bea}
		\end{minipage}}
	\end{table*}	

Table \ref{bea} reports the performance of TDA-based sparse index tracking portfolios $C_1$ to $C_6$ on the eight indices and the fully-replicating index tracking portfolio, $TE_{all}$. There are several findings. First, our $C_2$ portfolio achieves the lowest tracking error when contrasted against the other TDA-based portfolios for all the indices under consideration. What's more, it is striking that the tracking error stemming from $C_2$ is statistically insignificant compared to the tracking error emanating from the $TE_{all}$ portfolio for five of the eight indices. Second, the correlation between the $C_2$ portfolios and the respective market index is notably higher than that observed in other TDA-based portfolios for five of the eight indices. Even for the index where the correlation from the $C_2$ portfolios is not the highest, it remains remarkably close to the highest correlation value. Third, although the portfolios generate a negative excess mean return for some of the indices, it is interesting that even in these instances, the $C_2$ portfolios consistently outperform their counterparts, namely the $C_5$ and $C_6$ portfolios, which are constructed using conventional similarity matrices. Fourth, it is remarkable that the $C_2$ portfolios are well-diversified, as indicated by the Herfindahl-Hirschman Index (HHI), with $C_2$ portfolios achieving the lowest values among the portfolios. It is essential to underscore that the lower values of HHI for $C_2$ attributes to the specific assets chosen based on the clustering with a similarity measure based on the kernel $K_2$ and is not merely a by-product of the higher value of $M$ (number of assets selected) as evidenced by the HHI values for T5. These findings indicate that the based sparse index tracking portfolio, $C_2$, consistently outperforms the other index tracking portfolios under consideration. 


\begin{rem}
We also compare the performance of our best-performing TDA-based sparse portfolio, $C_2$, to the index tracking portfolio proposed in \cite{beasley2003evolutionary}, which we denote by $C_{beas}$. It is worth noting that our TDA-based sparse portfolio, unlike $C_{beas}$, doesn't require the number of assets to be included in the portfolio. Since the composition of the optimal index tracking portfolios in \cite{beasley2003evolutionary} is not publicly available, we directly use the tracking error values reported in that paper. Further, the tracking error used in \cite{beasley2003evolutionary} is 
$$ TE_{beas}= \frac{1}{T} \sqrt{\sum_{j=1}^{T}(R_j(\bm w)-R_{0,j})^2 } . \label{te}$$
Therefore, for a fair comparison, we also compute the tracking error, $TE_{beas}$, for our portfolio, $C_2$. Table \ref{beascomp} presents the out-of-sample $TE_{beas}$ for $C_2$ and $C_{beas}$ for data sets T1-T5, the indices used in \cite{beasley2003evolutionary}. Interestingly, our $C_2$ portfolio outperforms $C_{beas}$ for all except T4 indices in terms of the tracking error $T_{beas}$.

\begin{table}[t!]
\centering
  \caption{Out-of-sample value of tracking error $TE_{beas}$ for $C_2$ portfolio and $C_{beas}$ portfolio from $Beasley_{IT}$ model. Best value is highlighted in \textbf{Bold}}
    \begin{tabular}{|l|rr|}
   \hline
  Data &{$C_2$} & { $C_{beas}$} \\
  \hline
  T1  &\textbf{2.96E-04} & 1.27E-03 \\
T2    &\textbf{8.77E-04} & 2.05E-03 \\
 T3   &\textbf{7.40E-04} & 9.58E-04 \\
T4   & {1.58E-03} & \textbf{1.03E-03} \\
  T5& \textbf{ 6.35E-04} & 8.21E-04 \\
  \hline
    \end{tabular}%
  \label{beascomp}%
\end{table}%
\end{rem}

The empirical study demonstrates the efficacy of the TDA-based kernel $K_2$ in facilitating the superior asset selection for the index-tracking model. Further, the kernel $K_2$, when integrated with affinity propagation clustering, results in an entirely data-driven approach, thus eliminating the need to pre-specify the number of assets to be included in the tracking portfolio. Therefore, our proposed TDA based strategy holds promise in enhancing asset allocation decisions, particularly in the context of index tracking.


\subsection{Strategy 2: Sparse Mean-variance and Global Minimum Variance Portfolios}\label{mean_variance}

Next, we assess the performance of our TDA-based data-driven mean-variance and global minimum variance optimal portfolios. More precisely, for a fixed window, we follow Strategy 2, i.e., given the similarity matrices $K_1$ to $K_6$, we perform affinity propagation clustering and select the exemplars, then construct the optimal portfolios solving problems (\ref{mvp}) and (\ref{gvp}) on the filtered stocks. We denote the optimal mean-variance and the global minimum variance portfolios corresponding to a similarity matrix $K_i$ by $MV_i$ and $GMV_i$, respectively, where $i=1,2,\ldots,6$. We also create the following portfolios from the literature as benchmarks for comparative analysis.
\begin{itemize}
\item[1.] $\text{MV and GMV with all assets}$: These portfolios correspond to the optimal mean-variance and global minimum variance portfolios created by solving  (\ref{mvp}) and (\ref{gvp}) with all of the assets that make up the benchmark index (without clustering). We denote these portfolios by $MV_A$ and $GMV_A$, respectively.

\item[2.] $\text{Naive portfolio:}$ This portfolio corresponds to the equal-weighted portfolio.
\end{itemize}

Before analyzing the performance of optimal portfolios, we first observe that the median number of assets (exemplars) for similarity matrices $K1-K6$ are 60, 64, 74, 86, 73, and 76, thus confirming that the resulting portfolios are indeed sparse.


Panel A in Table \ref{strategy2_outsample} shows the out-of-sample performance of the TDA-based sparse mean-variance portfolios. We observe that the portfolio $MV2$ achieves the maximum values for mean, SR, and CEQ, and the minimum value of TR, thus outperforming the other portfolios, including TDA-based portfolios and the benchmark portfolios, namely, $MVA$, the na\"{\i}ve strategy and the benchmark index. For instance, $MV2$ achieves the daily Sharpe ratio of .0733 (.3359 per month / 1.1636 per annum) compared to 0.688 and 0.521 achieved by $MVA$ and $SP500$. Further, the least value of the turnover ratio for $MV2$ indicates less re-balancing across the windows and the formation of stable portfolios. In contrast to the performance of TDA-based index-tracking portfolios, $MV4$ fails to consistently beat other TDA-based portfolios as well as the benchmark portfolios. Nonetheless, these findings confirm our conclusion that the TDA-based selection strategy, in particular, based on kernel $K_2$, outperforms the benchmark models, thus providing an excellent asset selection criteria.\footnote{We also observe that the $MVA$ portfolio has a higher out-of-sample mean, SR, and CEQ compared to those from the na\"{\i}ve $1/N$ portfolio. We also perform hypothesis tests to assess whether the values from $MVA$ and na\"{\i}ve strategy are statistically significant. We observed that the mean and CEQ from $MVA$ is significantly greater than that from na\"{\i}ve strategy at 99\% confidence level while the Sharpe ratios from the two are not significantly different, a similar conclusion was drawn in \cite{demiguel2009optimal}.} Panel B in Table \ref{strategy2_outsample} shows the out-of-sample performance of the TDA-based sparse global minimum variance portfolios. A similar conclusion can be drawn that in a $K_2$ based portfolio, $GV2$ outperforms other TDA-based portfolios, with the exception that $GV5$ has a lower turnover ratio than $GV2$, but the difference is not very large. Further, $GV2$ beats the benchmark portfolios $GVA$, the na\"{\i}ve, and the index in terms of mean and Sharpe ratio. These findings were further confirmed using the statistical tests for mean, SR, and CEQ; for instance, we observed that $MV2$ has a significantly higher mean, SR, and CEQ than $MVA$ portfolios at a 95\% confidence level.

\begin{table*}[t!]
\footnotesize 
  \caption{In this table, we report the out-of-sample mean return, Sharpe ratio (SR), certainty-equivalent (CEQ) return, and portfolio turnover of the optimal mean-variance ($MV$) and global minimum variance ($GMV$) portfolios. The portfolios $MV1-MV6 (GMV1-GMV6)$ correspond to the sparse optimal mean-variance (global minimum variance) portfolios obtained based on Strategy 2 corresponding to six similarity matrices. The portfolios $MVA$ and $GMVA$ correspond to the optimal mean-variance and global minimum variance portfolios obtained on all assets, and the naive corresponds to an equal-weighted portfolio. We calculate these measures using realized simple returns of the corresponding portfolios over the out-of-sample period. We also report these measures for the index for comparison. The best and second-best values are highlighted in bold and italics, respectively.}
\centering  \resizebox{.9\textwidth}{!}{\begin{minipage}{\textwidth}
    \begin{tabular}{l|rrrrrr|rrr}
\hline  
\multicolumn{9}{l}{A. Mean-Variance Portfolios}\\
     & \multicolumn{1}{l}{MV1} & \multicolumn{1}{l}{MV2} & \multicolumn{1}{l}{MV3} & \multicolumn{1}{l}{MV4} & \multicolumn{1}{l}{MV5} & \multicolumn{1}{l}{MV6} & \multicolumn{1}{|l}{MVA} & \multicolumn{1}{l}{Naïve} & \multicolumn{1}{l}{Index} \\
\hline
    Mean  & 9.69E-04 & \textbf{1.82E-03} & 1.22E-03 & 7.67E-04 & 7.27E-04 & 8.83E-04 & \textit{1.70E-03} & 5.89E-04 & 4.70E-04 \\
    SR    & 4.67E-02 & \textbf{7.33E-02} & 5.58E-02 & 3.28E-02 & 3.55E-02 & 4.19E-02 & \textit{6.88E-02} & 6.22E-02 & 5.21E-02 \\
        CEQ   & 7.54E-04 & \textbf{1.51E-03} & 9.79E-04 & 4.94E-04 & 5.17E-04 & 6.61E-04 & \textit{1.40E-03} & 5.44E-04 & 4.29E-04\\
\vspace{0.5cm}
    TR & 1.74E+00 & \textbf{1.37E+00} & 1.78E+00 & 1.67E+00 & 1.45E+00 & 1.39E+00&&& \\
\hline
\multicolumn{9}{l}{B. Global Minimum Variance Portfolios}\\
     & \multicolumn{1}{l}{GV1} & \multicolumn{1}{l}{GV2} & \multicolumn{1}{l}{GV3} & \multicolumn{1}{l}{GV4} & \multicolumn{1}{l}{GV5} & \multicolumn{1}{l}{GV6} & \multicolumn{1}{|l}{GVA} & \multicolumn{1}{l}{Naïve} & \multicolumn{1}{l}{Index} \\
\hline 
Mean  & 3.50E-04 & \textbf{4.42E-04} & 3.91E-04 & 3.44E-04 & 3.95E-04 & 3.35E-04 & \textit{3.96E-04} & 5.89E-04 & 4.70E-04 \\
    SR    & 5.04E-02 & \textbf{6.45E-02} & 5.75E-02 & 4.93E-02 & 5.73E-02 & 4.84E-02 & \textit{6.38E-02} & 6.22E-02 & 5.21E-02 \\
        CEQ   & 3.26E-04 & \textit{4.18E-04} & 3.68E-04 & 3.20E-04 & 3.72E-04 & 3.11E-04 & 3.77E-04 & \textbf{5.44E-04} & 4.29E-04\\
    TR & 1.46E+00 & 1.24E+00 & 1.70E+00 & 1.67E+00 & \textbf{1.22E+00} & 1.26E+00 &&&\\
\hline
    \end{tabular}%
  \label{strategy2_outsample}%
  \end{minipage}}
 \end{table*}%

To check our strategy's robustness, we study the in-sample and out-of-sample statistics by calculating the mean, SR, and CEQ values for each window (in-sample and out-of-sample) and then taking the average across the windows. Tables \ref{strategy2_inoutsample} presents the in-sample and out-of-sample statistics and their differences. We observe a performance degradation in the out-of-sample period, which is intuitive since the portfolios are optimized over the in-sample period and are tested in the out-of-sample period where no optimization occurs. Further, for $MV2$ and $GV2$, the difference between the in-sample and out-of-sample average values is less than the $MVA$ and $GVA$ portfolios, which suggests that TDA-based filtering strategy is more robust to the changing market dynamics.


\begin{table*}[t!]
\centering
\footnotesize  
\caption{Average (calculated across the windows) in-sample and out-of-sample values of Mean, SR and CEQ for the portfolios MV1-MV6 and GV1-GV6. The difference in the in-sample and out-of-sample values acts as an indicator for consistency and robustness of the portfolios. Minimum value of the difference is highlighted in \textbf{bold}}
  \resizebox{.9\textwidth}{!}{\begin{minipage}{\textwidth}
    \begin{tabular}{ll|rrrrrr|r}
\hline
\multicolumn{9}{l}{A. Mean-Variance Portfolios}\\
          && \multicolumn{1}{l}{MV1} & \multicolumn{1}{l}{MV2} & \multicolumn{1}{l}{MV3} & \multicolumn{1}{l}{MV4} & \multicolumn{1}{l}{MV5} & \multicolumn{1}{l}{MV6} & \multicolumn{1}{|l}{MVA} \\
\hline
    Mean & In & 3.521E-03 & 4.612E-03 & 3.825E-03 & 3.878E-03 & 3.742E-03 & 3.750E-03 & \textbf{5.530E-03} \\
    &Out & 9.689E-04 & \textbf{1.825E-03} & 1.217E-03 & 7.668E-04 & 7.269E-04 & 8.835E-04 & 1.704E-03 \\
\vspace{0.5cm}
    &Difference  & \textbf{2.552E-03} & 2.787E-03 & 2.608E-03 & 3.111E-03 & 3.015E-03 & 2.867E-03 & 3.826E-03 \\
    SR &In & 1.846E-01 & 1.806E-01 & 1.921E-01 & 1.711E-01 & 1.853E-01 & 1.880E-01 & \textbf{2.186E-01} \\
    &Out & 7.870E-02 & 8.542E-02 & 7.898E-02 & 4.842E-02 & 5.695E-02 & 7.132E-02 & \textbf{9.374E-02} \\
\vspace{0.5cm}
    &Difference  & 1.059E-01 & \textbf{9.519E-02} & 1.131E-01 & 1.227E-01 & 1.283E-01 & 1.167E-01 & 1.248E-01 \\
    CEQ &In & 3.310E-03 & 4.242E-03 & 3.592E-03 & 3.561E-03 & 3.498E-03 & 3.516E-03 & \textbf{5.158E-03} \\
    &Out & 7.517E-04 & \textbf{1.513E-03} & 9.741E-04 & 4.901E-04 & 5.137E-04 & 6.597E-04 & 1.394E-03 \\
\vspace{0.5cm}
    &Difference  & \textbf{2.558E-03} & 2.729E-03 & 2.618E-03 & 3.071E-03 & 2.984E-03 & 2.856E-03 & 3.765E-03 \\
\hline
\multicolumn{9}{l}{B. Global Minimum Variance Portfolios}\\

          && \multicolumn{1}{l}{GV1} & \multicolumn{1}{l}{GV2} & \multicolumn{1}{l}{GV3} & \multicolumn{1}{l}{GV4} & \multicolumn{1}{l}{GV5} & \multicolumn{1}{l}{GV6} & \multicolumn{1}{|l}{GVA} \\
          \hline
    Mean &In & 5.137E-04 & 5.227E-04 & 5.067E-04 & \textbf{5.289E-04} & 5.276E-04 & 4.861E-04 & 5.017E-04 \\
    &Out & 3.504E-04 & \textbf{4.419E-04} & 3.911E-04 & 3.443E-04 & 3.955E-04 & 3.355E-04 & 3.961E-04 \\
\vspace{0.5cm}
    &Difference  & 1.633E-04 & \textbf{8.084E-05} & 1.156E-04 & 1.846E-04 & 1.321E-04 & 1.506E-04 & 1.055E-04 \\
    SR &In & 1.009E-01 & \textbf{1.044E-01} & 1.040E-01 & 9.789E-02 & 1.041E-01 & 9.965E-02 & 1.274E-01 \\
    &Out & 7.462E-02 & \textbf{9.952E-02} & 8.179E-02 & 7.416E-02 & 8.972E-02 & 8.321E-02 & 9.334E-02 \\
\vspace{0.5cm}
    &Difference  & 2.628E-02 & \textbf{4.830E-03} & 2.219E-02 & 2.373E-02 & 1.437E-02 & 1.643E-02 & 3.409E-02 \\
    CEQ &In & 4.961E-04 & 5.052E-04 & 4.902E-04 & 5.092E-04 & \textbf{5.102E-04} & 4.689E-04 & 4.907E-04 \\
    &Out & 3.260E-04 & \textbf{4.181E-04} & 3.676E-04 & 3.195E-04 & 3.714E-04 & 3.111E-04 & 3.765E-04 \\
\vspace{0.5cm}
    &Difference  & 1.702E-04 & \textbf{8.713E-05} & 1.226E-04 & 1.897E-04 & 1.388E-04 & 1.578E-04 & 1.142E-04 \\
\hline
    \end{tabular}%
     \label{strategy2_inoutsample}%
    \end{minipage}}
 \end{table*}%

\textbf{COVID-19 Robustness Check} Next, we analyze the portfolio performance in the more volatile period from August 2019 to August 2020, which includes the COVID-19 pandemic. Table \ref{strategy2_robustness} presents the out-of-sample performance. Panel A clearly shows that the sparse TDA-based portfolio $MV2$ performs really well even in this highly volatile period. More specifically, the TDA-based portfolio $MV_2$ outperforms the $MVA$, the naive portfolio, and the benchmark index. This is demonstrated with significant gaps; for instance, the Sharpe ratio of 1.38E-01 achieved by $MV_2$ is better and statistically significantly than 1.23E-02 and 1.97E-02 achieved by $MVA$ and the naive portfolio, respectively. Similar to the TDA-based mean-variance portfolios, there is clear pattern that emerges in the performance of TDA-based global minimum variance portfolios. That is, $GV_2$ outperform $GVA$ and na\"ive portfolio, thus confirming our findings that TDA-based portfolios, in particular, based on the kernel $K_2$ are not only sparse but also perform superior to traditional models even in turbulent market conditions, safeguarding against significant losses. These results underline the out-of-sample efficacy of TDA-based models in managing risk and delivering returns.

\begin{table*}[ht!]
\footnotesize 
\centering
  \caption{Out-of-sample performance analysis of portfolios $MV1-MV6, MVA, GV1-GV6, GVA$, na\"{\i}ve strategy and the index on COVID-19 data set. The best and second-best values are highlighted in bold and italics, respectively.}
  \resizebox{.8\textwidth}{!}{\begin{minipage}{\textwidth}
    \begin{tabular}{l|rrrrrr|rrr}
\hline
\multicolumn{9}{l}{A. Mean-Variance Portfolios}\\
     & \multicolumn{1}{l}{MV1} & \multicolumn{1}{l}{MV2} & \multicolumn{1}{l}{MV3} & \multicolumn{1}{l}{MV4} & \multicolumn{1}{l}{MV5} & \multicolumn{1}{l}{MV6} & \multicolumn{1}{|l}{MVA} & \multicolumn{1}{l}{Naïve} & \multicolumn{1}{l}{Index} \\
\hline
Mean& 5.39E-03 & \textbf{6.76E-03 }& \textit{6.00E-03} & 4.93E-05 & 4.59E-03 & 4.04E-03 & {1.15E-03} & 5.34E-04 & 5.96E-04 \\

SR & 8.17E-02 & \textbf{1.38E-01} & \textit{1.21E-01} & 7.76E-04 & 7.84E-02 & 6.56E-02 & 1.23E-02 & 1.97E-02 & 2.43E-02 \\

CEQ & 3.22E-03 & \textbf{1.59E-02} & \textit{7.77E-03} & -1.97E-03 & 2.88E-03 & 2.14E-03 & 6.36E-04& 1.66E-04 & 2.96E-04 \\
\vspace{0.5cm}
TR & 1.98 & \textbf{1.14} & 1.91 & 1.56 & \textit{1.48}& 1.49 &  &       &  \\
\hline
\multicolumn{9}{l}{B. Global Minimum Variance Portfolios}\\
     & \multicolumn{1}{l}{GV1} & \multicolumn{1}{l}{GV2} & \multicolumn{1}{l}{GV3} & \multicolumn{1}{l}{GV4} & \multicolumn{1}{l}{GV5} & \multicolumn{1}{l}{GV6} & \multicolumn{1}{|l}{GVA} & \multicolumn{1}{l}{Naïve} & \multicolumn{1}{l}{Index} \\
\hline 
    Mean & -9.35E-04 & \textbf{4.44E-04} & -2.22E-04 & \textit{2.71E-04} & -1.16E-03 & -1.98E-04 & \textit{5.75E-04} & 5.34E-04 & 5.96E-04 \\

SR& -4.30E-02 & \textbf{1.96E-02} & -7.49E-03 & \textit{1.29E-02 }& -5.05E-02 & -8.29E-03 & \textit{2.91E-02} & 1.97E-02 & 2.43E-02 \\

CEQ & -1.17E-03 & \textbf{1.87E-04} & -6.60E-04 & \textit{4.94E-05 }& -1.42E-03 & -4.82E-04 & \textit{3.80E-04} & 1.66E-04 & 2.96E-04 \\
  TR &\textit{ 1.40} & \textbf{1.17 }& 1.71 & 1.47 & 1.40& 1.40 &  &       &  \\
\hline
    \end{tabular}%
  \label{strategy2_robustness}%
\end{minipage}}
 \end{table*}%

\newpage

\subsection{Sensitivity to Clustering Algorithm}\label{appendix_comparison}

 As a robustness check, we replace the APC with two alternative algorithms: $K-$medoids and Hierarchical clustering to demonstrate the effectiveness of TDA based distances, regardless of the employed clustering algorithm.\footnote{ We also tried Density-Based Spatial Clustering of Applications with Noise (DBSCAN); however, given the nature of our data, most points were classified as noise rather than clusters. Therefore, we dropped it from further analysis.} To do this, we study the performance of portfolios from Strategy 1 and Strategy 2 over the period September 2020 to August 2022. We use the rolling window approach with a window size of 6 months, shifting by one month, resulting into 18 windows. For $K-$medoids, we obtained the optimal number of clusters by maximizing the average silhouette score across $K=[10, 50]$. For consistency, we used the same number of clusters, obtained in $K$-medoids, in Hierarchical clustering. In Hierarchical clustering, we obtained centroids (selected as the point with minimum total distance to other points in the cluster) as the cluster representative for Strategy 2. Table \ref{st1_comp} presents the performance of the index tracking portfolios (Strategy 1), while Table \ref{st2_comp} present the performance of mean-variance portfolios and global minimum variance portfolios (Strategy 2).

\begin{table}[t!] 
  \centering
  \caption{Out-of-sample performance of clustering based portfolios $C_1-C_6$ during September 2020 to August 2022 for different clustering algorithms.}
    \begin{tabular}{l|rrrrrr}
    \hline
  
     & \multicolumn{1}{l}{$C_1$} & \multicolumn{1}{l}{$C_2$} & \multicolumn{1}{l}{$C_3$} & \multicolumn{1}{l}{$C_4$} & \multicolumn{1}{l}{$C_5$} & \multicolumn{1}{l}{$C_6$} \\
     \hline
       & \multicolumn{6}{c}{Affinity propagation clustering}\\
       \hline
TE           & 5.88E-03 & \textit{3.59E-03} & 7.82E-03 & \textbf{2.43E-03} & 5.38E-03 & 4.98E-03 \\
EMR          & -3.76E-04 & \textbf{1.10E-04} & -8.33E-05 & \textit{9.17E-05} & 3.27E-06 & 1.15E-05 \\
COR          & 0.875 & \textit{0.952} & 0.748 & \textbf{0.979} & 0.928 & 0.937 \\
INFO         & -6.40E-02 & \textit{3.07E-02} & -1.06E-02 & \textbf{3.77E-02} & 6.08E-04 & 2.31E-03 \\
TR           & \textbf{0.344} & 1.10 & 1.36 & 1.46 & 1.25 & \textit{0.911} \\
\hline
  & \multicolumn{6}{c}{K-Medoids clustering}\\
       \hline
TE           & 5.62E-03 & \textit{3.20E-03} & 4.83E-03 & \textbf{2.29E-03} & 4.47E-03 & 4.24E-03 \\
EMR          & -4.01E-04 & \textit{1.89E-04} & 2.91E-05 & 1.83E-04 & \textbf{1.93E-04} & 1.05E-04 \\
COR          & 0.885 & \textit{0.962} & 0.903 & \textbf{0.982} & 0.947 & 0.952 \\
INFO         & -7.14E-02 & \textit{5.91E-02} & 6.03E-03 & \textbf{8.01E-02} & 4.31E-02 & 2.48E-02 \\
TR           & \textbf{0.679} & 1.36 & 2.00 & 1.68 & 1.49 & \textit{1.35} \\
\hline
  & \multicolumn{6}{c}{Hierarchical clustering}\\
       \hline
TE           & 5.89E-03 &\textbf{ 4.85E-03} & 5.98E-03 & 7.33E-03 & \textit{5.32E-03} & 6.01E-03 \\
EMR          & -2.82E-04 & 2.30E-05 & \textbf{3.71E-04} & -5.21E-04 & -6.56E-05 & \textit{2.46E-04} \\
COR          & 0.869 & \textit{0.913} & 0.850 & 0.842 & \textbf{0.926} & 0.909 \\
INFO         & -4.79E-02 & 4.75E-03 & \textbf{6.20E-02} & -7.10E-02 & -1.23E-02 & \textit{4.10E-02} \\
TR           & \textbf{1.33} & \textit{2.27} & 2.45 & 3.18 & 2.32 & 2.29 \\
\hline
    \end{tabular}\label{st1_comp}
\end{table}

From Table \ref{st1_comp}, we observe that the TDA based portfolios emerge as the best-performing portfolios across all clustering methods. In particular, $C_2$ and $C_4$ are top-performing portfolios in affinity propagation clustering and $K$-medoids clustering, while $C_3$ performs well under Hierarchical clustering. For instance, $C_4$ attains minimum TE, maximum correlation and maximum information ratio while $C_2$ attains the second position in terms of these performance measures in APC and $K$-medoids clustering. The performance of these portfolios is mixed under Hierarchical clustering, with no clear standout. Nonetheless, $C_3$ outperforms other algorithms in terms of EMR and INFO. Overall, the findings from Table \ref{st1_comp} align with our previous findings that TDA based portfolios outperform other benchmarks. Additionally, while all clustering methods show high correlation (COR) with the benchmark, APC achieves this alignment without excessively high turnover, striking a balance between tracking the benchmark and maintaining low trading costs. The information ratio for APC portfolios is also moderate but stable, indicating a reasonable trade-off between excess return and risk. K-medoids and hierarchical clustering may yield higher INFO values in isolated cases, but their elevated turnover suggests that these risk-adjusted returns may not be sustainable due to the increased transaction costs. From Table \ref{st2_comp}, we observe that no consistent pattern emerges across clustering algorithms. For instance, MV2 with APC dominates other portfolios, as also observed in our main analysis, while MV6 and MV5 dominates TDA based portfolios with $K$-medoids and hierarchical clustering respectively. In contrast, there is no clear standout performing kernel in GMV portfolios regardless of the clustering algorithm. While APC shows slightly better results for GV4, $K$-medoids works well for GV5 and hierarchical clustering is better for GV1. Similar to Table \ref{st1_comp}, APC exhibits lower turnover values compared to other clustering algorithms, indicating potentially lower transaction costs. In conclusion, TDA based portfolios outperforms the benchmark across all clustering methods, particularly excelling in affinity propagation clustering.

\begin{table}
  \centering
  \caption{Out-of-sample performance of clustering based portfolios $MV1-MV6$ and $GV1-GV6$ during September 2020 to August 2022 for different clustering algorithms.}
\begin{tabular}{l|rrrrrr}
\hline  
& \multicolumn{5}{l}{Panel A. Mean-Variance Portfolios}\\
     & \multicolumn{1}{l}{MV1} & \multicolumn{1}{l}{MV2} & \multicolumn{1}{l}{MV3} & \multicolumn{1}{l}{MV4} & \multicolumn{1}{l}{MV5} & \multicolumn{1}{l}{MV6} \\
\hline
   & \multicolumn{6}{c}{Affinity propagation clustering}\\
 \hline
Mean       & 2.31E-04 & \textbf{1.84E-03} & \textit{1.25E-03} & 8.28E-04 & 9.94E-05 & 7.83E-04 \\
SR         & 1.13E-02 & \textbf{5.84E-02} & \textit{4.77E-02} & 3.95E-02 & 3.93E-03 & 2.49E-02 \\
CEQ        & 2.05E-05 & \textbf{1.34E-03} &\textit{ 9.07E-04 }& 6.08E-04 & -2.20E-04 & 2.88E-04 \\
TR         & \textit{1.529} & 1.974 & 1.930 & 1.749 & 1.573 & \textbf{1.294} \\
\hline
  & \multicolumn{6}{c}{K-Medoids clustering}\\
\hline
Mean       & -3.33E-05  & 4.10E-04 & 4.78E-04 & \textit{1.11E-03} & 8.92E-04 & \textbf{1.34E-03} \\
SR         & -1.21E-03 & 1.72E-02 & 1.74E-02 & \textit{4.42E-02} & 3.02E-02 & \textbf{5.38E-02} \\
CEQ        & -4.10E-04 & 1.25E-04 & 1.01E-04 & \textit{7.95E-04} & 4.55E-04 & \textbf{1.03E-03} \\
TR         & 2.359 & 2.800 & 3.536 & 2.337 & \textbf{2.122} & \textit{2.300} \\
\hline
  & \multicolumn{6}{c}{Hierarchical clustering}\\
\hline
Mean       & -1.91E-04 & 7.79E-05  & 1.21E-03 & \textit{1.66E-03} & \textbf{1.77E-03} & 4.79E-04 \\
SR         & -5.93E-03 & 2.63E-03 & 4.19E-02 & \textit{5.18E-02} & \textbf{6.59E-02} & 1.56E-02 \\
CEQ        & -7.11E-04 & -3.59E-04 & \textit{7.94E-04} & \textbf{1.14E-03} & \textbf{1.41E-03} & 9.60E-06 \\
TR         & \textbf{2.889} & 3.458 & 4.952 & 3.546 & 3.261 & \textit{3.084} \\
\hline
& \multicolumn{5}{l}{Panel B. Global Minimum-Variance Portfolios}\\
   & \multicolumn{1}{l}{GV1} & \multicolumn{1}{l}{GV2} & \multicolumn{1}{l}{GV3} & \multicolumn{1}{l}{GV4} & \multicolumn{1}{l}{GV5} & \multicolumn{1}{l}{GV6} \\
\hline
   & \multicolumn{6}{c}{Affinity propagation clustering}\\
 \hline
Mean       & 1.66E-04 & 1.43E-04 & 1.31E-04 & \textbf{4.12E-04} & 1.60E-04 & \textit{3.62E-04} \\
SR         & 1.96E-02 & 1.64E-02 & 1.63E-02 & \textit{4.27E-02} & 1.01E-02 & \textbf{4.51E-02} \\
CEQ        & 1.30E-04 & 1.05E-04 & 9.85E-05 & \textbf{3.66E-04} & 1.01E-04 & \textit{3.30E-04} \\
TR         & 1.506 & 1.457 & 1.822 & 1.500 & \textbf{1.213} & \textit{1.226} \\
\hline
  & \multicolumn{6}{c}{K-Medoids clustering}\\
  \hline
Mean       & 2.08E-04 & -1.39E-04 & \textit{5.39E-04} & 2.57E-04 & \textbf{6.86E-04} & 3.38E-04 \\
SR         & 2.37E-02 & -1.49E-02 & \textit{6.01E-02} & 2.73E-02 & \textbf{8.23E-02} & 3.89E-02 \\
CEQ        & 1.69E-04 & -1.83E-04 & \textit{4.99E-04} & 2.13E-04 & \textbf{6.52E-04} & 3.00E-04 \\
TR         & 2.119 & 2.280 & 3.449 & 2.130 & \textbf{1.934} & \textit{1.942} \\
\hline
  & \multicolumn{6}{c}{Hierarchical clustering}\\
\hline
Mean       & \textbf{4.16E-04} & -1.41E-04 & -1.75E-04 & \textit{3.86E-04} & 3.43E-04 & 3.47E-04 \\
SR         & \textbf{4.98E-02} & -1.52E-02 & -1.95E-02 & \textit{4.37E-02} & 3.89E-02 & 4.09E-02 \\
CEQ        & \textbf{3.81E-04} & -1.84E-04 & -2.16E-04 & \textit{3.47E-04} & 3.04E-04 & 3.11E-04 \\
TR         & 2.681 & 2.992 & 4.373 & 3.346 & \textit{2.589} & \textbf{2.493} \\
\hline
\end{tabular}\label{st2_comp}

\end{table}

\section{Conclusion}\label{conclusion} 

In this paper, we introduced new distance measures on the space of persistence diagrams and landscapes.\ Notably, these distance measures incorporate the essential time component of a time series- a facet often overlooked by traditional distance measures. Employing these novel measures as an input to the affinity propagation clustering and proposed two strategies for asset selection to promote sparsity in the portfolio. Specifically, we presented a novel, entirely data-driven approach based on topological data analysis for selecting sparse portfolios in the context of index tracking and Markowitz portfolios. Our empirical findings demonstrate that the performance of sparse portfolios selected based on proposed novel distance measures is not merely comparable, but often excels in volatile market conditions relative to the prevailing methods across diverse set of performance measures. Thus, our paper contributes to sparse portfolio selection and applications of topological data analysis in finance.

\noindent
\textbf{Declaration:} There are no financial or non-financial competing interests.

\noindent
\textbf{Funding:} No funding was received.

\bibliographystyle{plain}
\bibliography{QF_format/clustering_bib}

\singlespacing

\appendix

\renewcommand{\thesubsection}{\Alph{section}.\arabic{subsection}}
\setcounter{table}{0}
\setcounter{figure}{0}
\renewcommand{\thetable}{A.\arabic{table}}
\renewcommand{\thefigure}{A.\arabic{figure}}

\newtheorem{thm}{Theorem}[section]
\newtheorem{theorem_app}[thm]{Theorem}
\section{Affinity Propagation Clustering}

The Affinity Propagation (AP) algorithm follows an iterative process that recursively searches for clusters by exchanging real-valued messages between data points, eventually identifying high-quality exemplars that define the resulting clusters. The AP algorithm starts with a similarity measure \( s(i, k) \) between sample points \( x_i \) and \( x_k \), and begins by treating each data point as a potential candidate for cluster representative, and finds a collection of optimal class representative points based on the preferences, \( s(k, k) \) with a larger value of \( s(k, k) \) denoting the high likelihood of selecting the sample point \( x_k \) as class representatives. Initially, AP considers each sample point with equal probability, i.e., giving \( s(k, k) \), $\forall k$, the same preference value \( P \).\footnote{We chose \( P \) as the median of the similarity matrix, a common choice in the literature, as it yields a moderate number of clusters. Alternative, selecting the minimum of similarity matrix results in a fewer clusters.} The numbers of exemplars are influenced by the preferences as well as the message passing procedure.

To obtain clusters, the AP algorithm primarily exchanges two types of messages: responsibility $r(i,k)$ and availability $a(i,k)$, between sample points, where $r(i,k)$ is sent from data point $x_i$ to potential exemplar $x_k$ depicting the extent to which $x_i$ favors $x_k$ as exemplar, and $a(i,k)$ indicates the suitability of selecting $x_k$ as a exemplar for $x_i$. The algorithm iteratively exchange these messages by updating $r$ and $a$ as follows, eventually determining $k$ as the final cluster center based on the decision matrix \( E \),
\[
r(i, k) \leftarrow s(i, k) - \max_{k' \neq k} \{ a(i, k') + s(i, k') \}
\]
\[
a(i, k) \leftarrow 
\begin{cases} 
\min \left(0, r(k, k) + \sum_{i' \neq \{i, k\}} \max(0, r(i', k)) \right) & \text{if } i \neq k \\
\sum_{i' \neq k} \max(0, r(i', k)) & \text{if } i = k 
\end{cases}
\]
where, $a(i, k)$ is set to be zero initially. A damping coefficient, $\lambda \in [0,1]$, is also used to avoid numerical oscillations and ensure stability and speedy convergence such that
\[
r^{(t+1)}(i, k) \leftarrow (1 - \lambda) r^{(t+1)}(i, k) + \lambda r^{(t)}(i, k)
\]
\[
a^{(t+1)}(i, k) \leftarrow (1 - \lambda) a^{(t+1)}(i, k) + \lambda a^{(t)}(i, k).
\]
That is, in each iteration, each message takes a weighted combination of its value from previous iteration and the current value, and update all responsibilities given availabilities, availabilities given responsibilities and finally combining them to monitor the exemplar decisions and concluding the AP clustering when these decisions did not change for certain number of iterations or a maximum number of iterations have been executed.

\subsection{Case study: Time series clustering}

In this section, we demonstrate the superior performance of the newly introduced distance measures by applying clustering on a standard data set from the machine learning literature. More precisely, we show that the distance measures $AWD$ and $DWD$ that quantify the dissimilarity between persistence diagrams $\Dcal_{x}$ and $\Dcal_{y}$ outperform the conventional Wasserstein distance ($WD$). Furthermore, we show that the proposed distance measure on landscapes $ALD$ and $DLD$ outperform the traditional distance between landscapes ($LD$), defined as $||\lambda_x-\lambda_y||_p$ where $\lambda_x$ and $\lambda_y$ are landscapes corresponding to the PD $\Dcal_{x}$ and $\Dcal_y$. For comparison purposes, we also compare the traditional distances, denoted by $d_1$ and $d_2$, defined in terms of Spearman's rank correlation $\rho_1$ and Pearson's correlation coefficient $\rho_2$, as follows
$d_i(x,y) =\sqrt{2(1-\rho_i(x,y))}~~i\in\{1,2\}.$ Finally, for completeness, we also consider the Euclidean distance $ED$. We evaluate the clustering performance in terms of the accuracy ($ACC$) \cite{he2019robust}, which ranges from 0 to 1, with a higher value indicating better clustering results.

We collect the data set from the UCI machine learning repository that contains 600 examples of control charts with 60 instances each, synthetically generated by the process in  \cite{alcock1999time}. The data consists of six different classes of control charts: Normal and  cyclic, increasing and decreasing trend, upward  and downward shift.

\begin{table}[ht!]
\centering
\caption{ Accuracy of clustering results with different distance metrics. Each value $v$ is to be read as $v \times 10^{-2}$}
\scalebox{0.8}{\begin{tabular}{l|r|r}
\hline
          & \multicolumn{2}{l}{ACC} \\
\hline
        Distance metric & K-med &  APC with Gaussian kernel\\
\hline
WD & 48.3 & 55\\
LD& 36.67 & 28.3\\
AWD& \textbf{60} & \textbf{56.67}\\
ALD& 35 & 45\\
DWD& \textbf{60} & \textbf{61.67}\\
DLD& 31.67 & 33.33\\
$d_1$& 33.3& 15\\
$d_2$& 31.67 & 15\\
ED& 51.67 & 15 \\
           \end{tabular}}%
  \label{clustering_test}%
\end{table}

Table \ref{clustering_test} presents the clustering results based on two standard algorithms, the affinity propagation clustering (APC) with $K$ clusters and the $K$-medoids with the distance measures mentioned above.\footnote{We use Gaussian exponential kernel as similarity measure defined using these distance measures as an input for APC. We fix embedding dimension $d=2$, and we take $\tau=1$ with $p=1$ for calculating the TDA-based distances. Moreover, we focus on the 1-dimensional features, i.e., loops. Further, the hyper-parameter $\sigma^2$ in the kernels is taken to be .01.} We set $K$ equal to six for illustration purposes. We observe that the TDA-based distance measures AWD and DWD outperform the other distance measures, particularly the existing Wasserstein distance. We also compare the accuracy values reported for the same data set in two recent papers \cite{cheng2020dense} and \cite{he2019robust}. The authors proposed a new clustering algorithm and a more accurate distance metric in both papers and compared the performance with different clustering methods. The accuracy values obtained using TDA-based distance measures DWD with affinity propagation clustering and AWD and DWD with K-medoids clustering are higher than the accuracy values reported in \cite{cheng2020dense} which is .58 and in \cite{he2019robust} which is .5886.

\section{Sensitivity to window length}

We also analyze the sensitivity of the clustering structure to the variations in the size of input window length, across different similarity matrices. We conduct this sensitivity analysis over two consecutive periods: September 2020 to August 2021 and September 2021 to August 2022. The study across two different years allows us to evaluate how changes in the input window length influence the clustering structure and assess the stability of the clustering outcomes across different timeframes. In each period, we construct similarity matrices and hence the clusters for varying input window lengths, starting with one month and expanding upto 12 months. For instance, for the first period, we obtain a similarity matrix using data from the month September 2020. We then extend the input window length incrementally: first to cover two months September 2020 to October 2020, then to three months and finally the period of 12 months, September 2020 to August 2021. We repeat this exercise for the second period, September 2021 to August 2022.

\begin{figure}[h!]
\begin{center}
		\includegraphics[scale=.55]{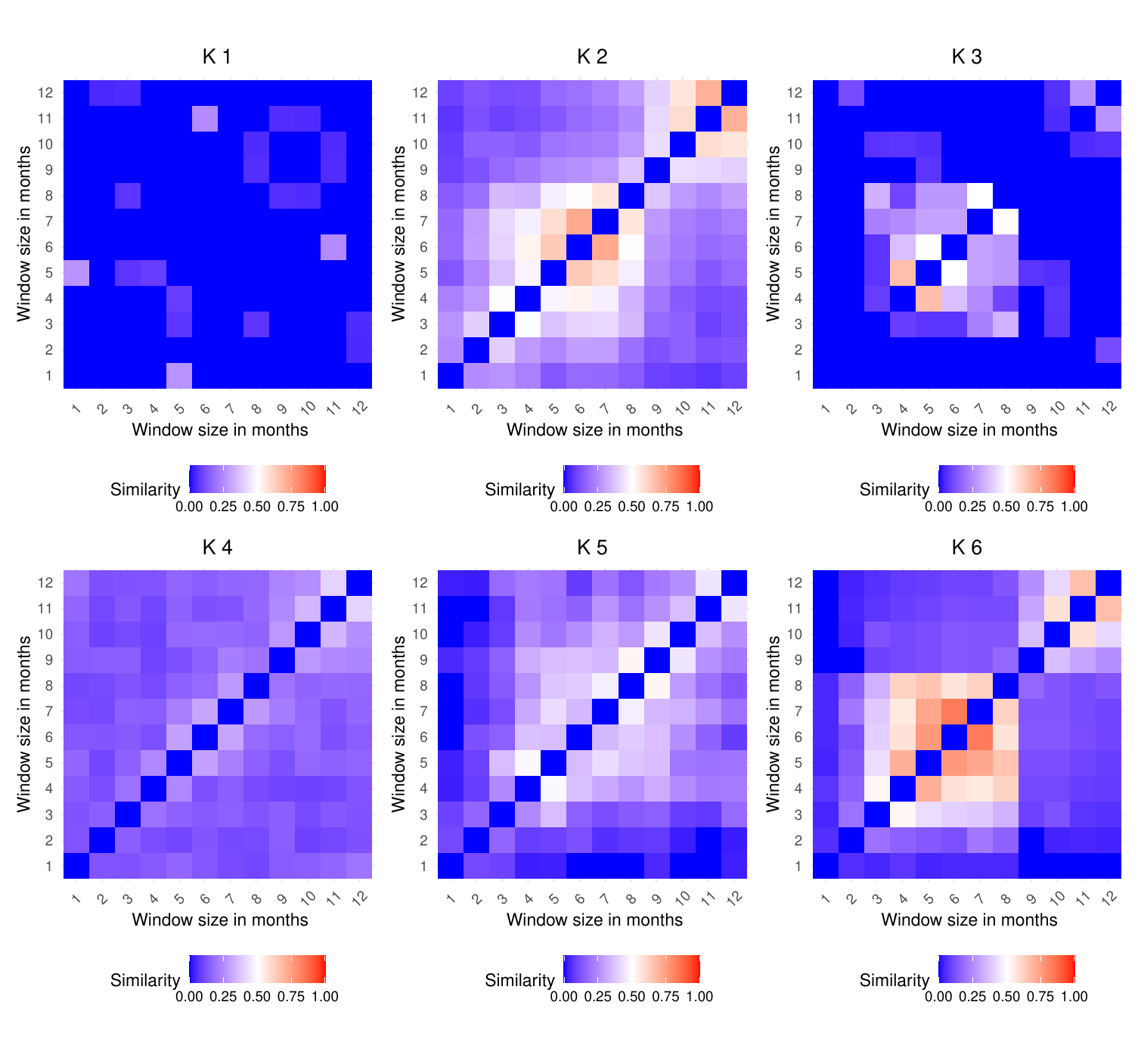}
    	
     \caption{Similarity between the cluster with benchmark across different window sizes with similarity matrices $K_1-K_6$ representing the sensitivity of the clustering with respect to the window size (Period - September 2020 to August 2021 ). Note that we have replaced the diagonal entries, which represent self-similarity and are typically 1, with 0 to enhance the contrast, leading to a more visible similarity patterns off-diagonally.}\label{ws1}
\end{center}
\end{figure}

We analyze the evolution of the clustering structure as the window length increases by calculating the adjusted rand index (ARI) between two clustering structures for two different window lengths. We conduct this analysis for all the six kernels, K1 to K6. Resulting heatmaps are presented in Figures \ref{ws1} and \ref{ws2} for first and second periods respectively. Note that we have replaced the diagonal entries, which represent self-similarity and are typically 1, with 0 to enhance the contrast, leading to a more visible similarity patterns off-diagonally.
\begin{figure}[h!]
\begin{center}
		\includegraphics[scale=.55]{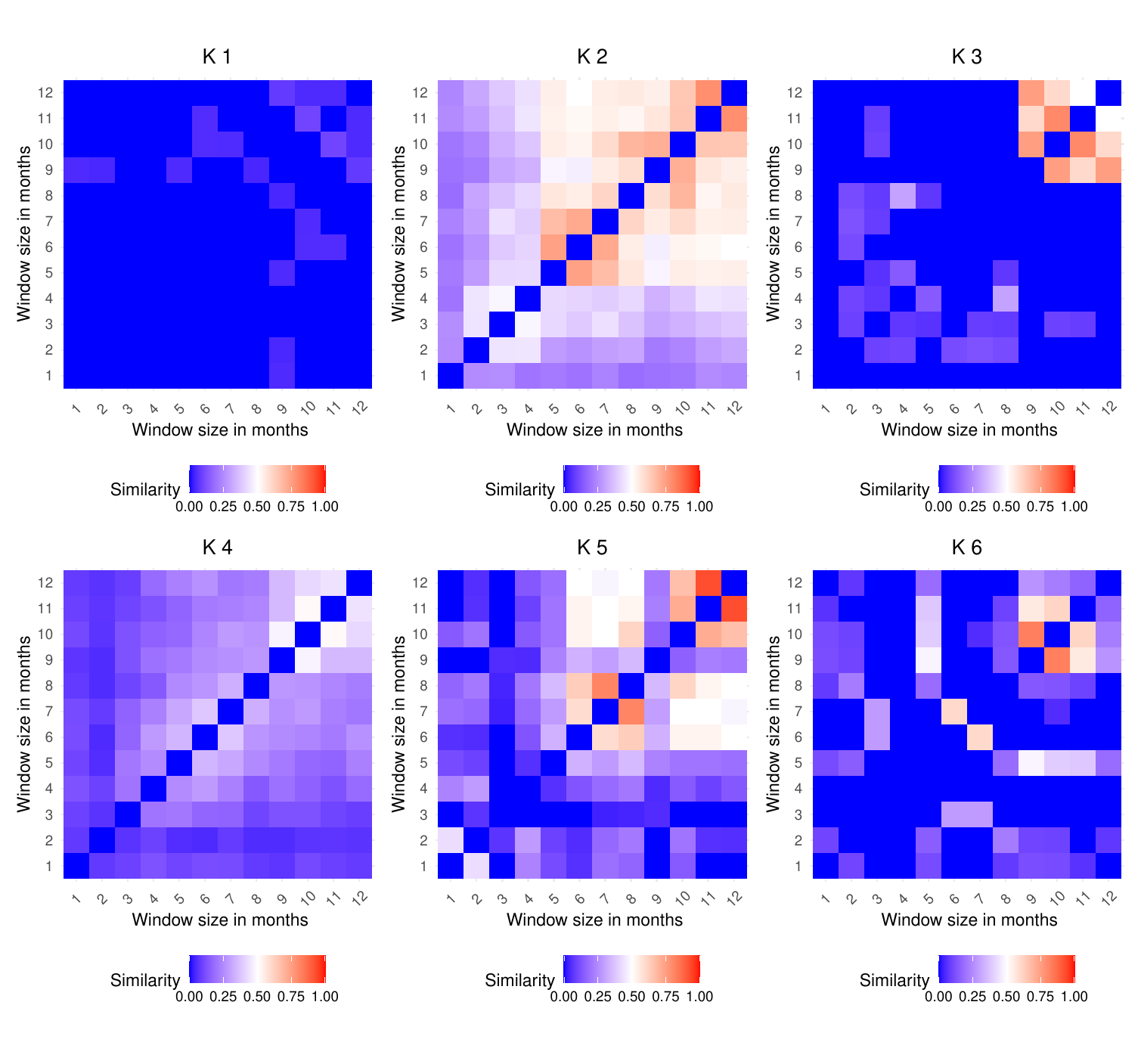}
    	
     \caption{Similarity between the cluster with benchmark across different window sizes with similarity matrices $K_1-K_6$ representing the sensitivity of the clustering with respect to the window size (Period - September 2021 to August 2022). Note that we have replaced the diagonal entries, which represent self-similarity and are typically 1, with 0 to enhance the contrast, leading to a more visible similarity patterns off-diagonally.}\label{ws2}
\end{center}
\end{figure}

Several observations emerge from these heat maps. First, the kernels $K_1$ and $K_3$ exhibit occasional patches across both the time-periods, indicating their high sensitivity to the window size and hence relative low stability over different input window lengths. Second, the kernel $K_4$ shows moderate similarity across varying window lengths, in particular, when window length is greater than 3 months, suggesting presence of robustness in kernel $K_4$. Third, the kernels $K_2$, $K_5$ and $K_6$, show high similarity in adjacent window lengths as well as when the difference in window size is large, signifying their robustness over broader range of the window size. Further, these kernels show two big patches (squares). The kernels $K_2$ and $K_6$ show similar pattern in the first period (Figure \ref{ws1}), while kernels $K_3$ and $K_5$ show similar patterns in the second period (Figure \ref{ws2}). It is remarkable that kernel $K_2$ not only captures what correlation based kernels capture, but also additional features, making it a more comprehensive measure. These findings show that the TDA based kernel $K_2$ and $K_4$ are more robust to the variation in the window length across different time-periods in contrast to correlation based distance measures that show period-dependent behavior.

\section{Analyzing the assets selected for index tracking}\label{similarity_based_selection}

It is essential to understand the characteristics of assets that are selected based on the similarity matrices $K_1-K_7$ for the index tracking problem. Specifically, we focus our study on the standard deviation, VaR, and CVaR of the difference return series, that is, the asset return series minus the index return series. Additionally, we analyze downside deviation from the index, and the correlation with the index.\footnote{VaR of the differenced series is the VaR relative to the index, which is the largest value by which the portfolio return can miss the index target in $1-\alpha$ fraction of cases. Similarly, CVaR relative to an index shows the mean deviation relative to the benchmark in the worst $\alpha$ cases.} Our analysis focuses on the top 60 stocks, identified as the most similar to the index, in each of the 119 windows for all seven similarity matrices. The heat maps in Figures (\ref{ht1}-\ref{ht5}) illustrate the average values of these measures across all windows, providing insights into the properties of the shortlisted stocks.


There are several observations. First, the assets selected using the similarity matrices $K_6$ and $K_5$ have a higher correlation than TDA-based similarity matrices, with $K_6$ achieving the highest correlation, which is not surprising given that the distance measure in this case is based on Pearson's correlation. On the other hand, among TDA-based similarity matrices, $K_2$ and $K_4$ filter assets have a higher correlation with the index than $K_1$, $K_3$. Second, in terms of standard deviation and downside deviation from the index, the similarity matrices $K_2$ and $K_7$ have the lowest values among all kernels. Third, the assets filtered using the similarity matrices $K_2$ and $K_7$ have the lowest values of the VaR and CVaR of the difference series (assets returns and index returns). Other kernels show varying behavior in different windows, but on average, assets picked using $K_3$, $K_5$, and $K_6$ have a higher tail deviation from the index. In conclusion, the assets selected using $K_2$ are most suitable for the index tracking problem, as further confirmed in Section \ref{empirical_results}. The robustness of the kernel $K_2$ is established from Figures \ref{cavht1}-\ref{cavht5}, which shows similar plots but for a randomly chosen window during the COVID-19 period. 

\begin{figure}[ht!]
		\begin{center}
			\subfigure[Correlation with index]{%
    \includegraphics[scale=0.52]{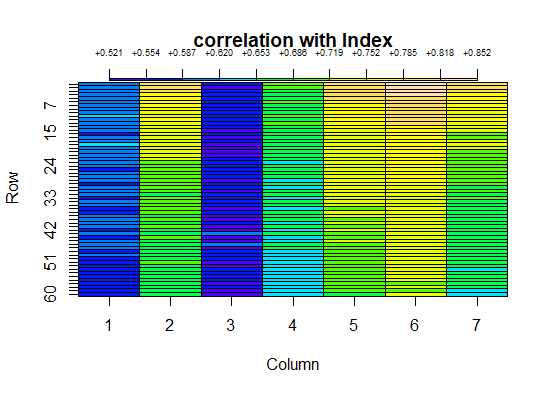}
    \label{ht1}}
\quad
			\subfigure[Standard deviation of difference series]{%
    \includegraphics[scale=0.52]{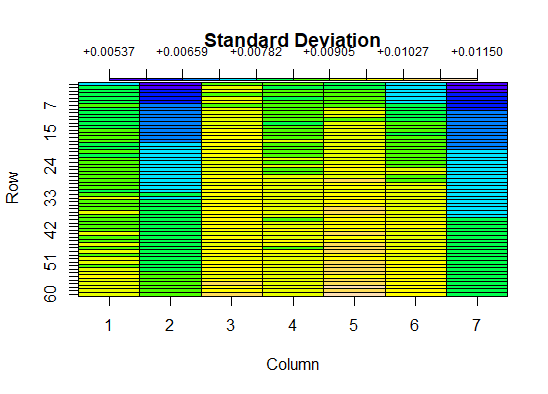}
    \label{ht2}}
\quad			
			\subfigure[Downside deviation of difference series]{%
        \includegraphics[scale=0.52]{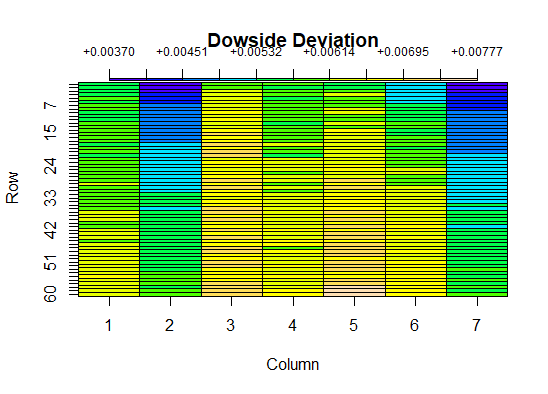}
    \label{ht3}}
\quad
			\subfigure[VaR of difference series]{%
    \includegraphics[scale=0.52]{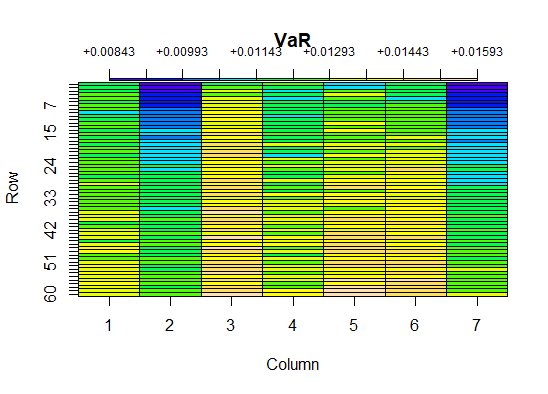}
    \label{ht4}}
\quad			
			\subfigure[CVaR of difference series]{%
\includegraphics[scale=0.52]{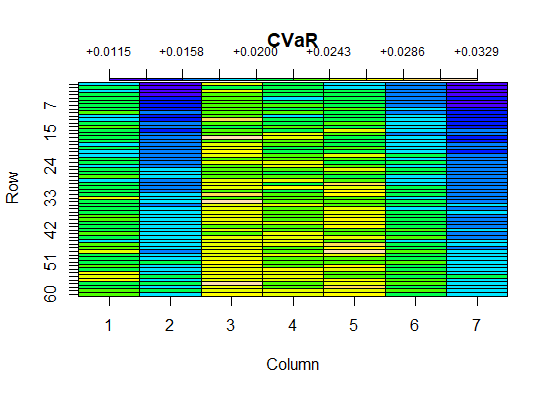}
    \label{ht5}}
   \end{center}
   \caption{Average values of five measures (correlation, standard deviation, downside deviation, VaR, CVaR) across all windows for the top 60 stocks (rows) and for the similarity matrices $K_1-K_7$ (columns)}
		\label{analyzing_assets}
\end{figure}

\begin{figure}[ht!]
		\begin{center}
			\subfigure[Correlation with index]{%
    \includegraphics[scale=0.44]{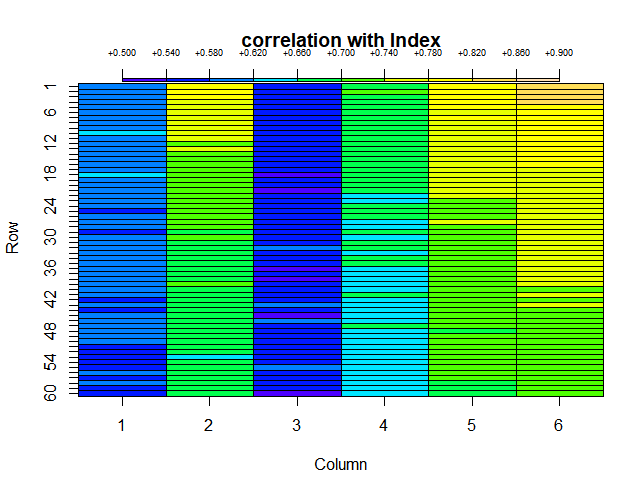}
    \label{cavht1}}
\quad
    \subfigure[Standard deviation of difference series]{
    \includegraphics[scale=0.44]{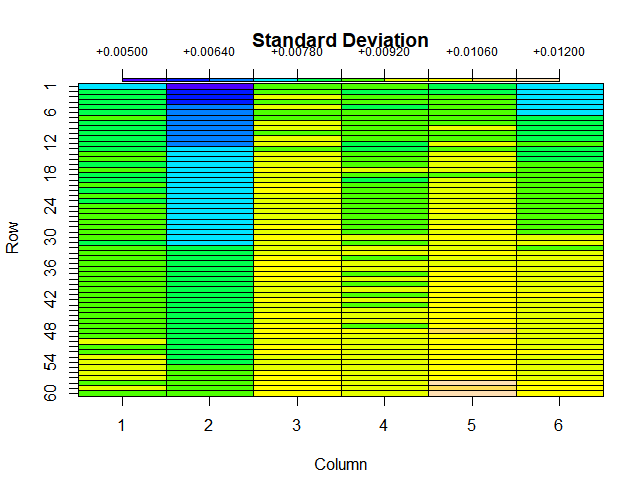}
    \label{cavht2}}
\quad
  \subfigure[Downside deviation of difference series]{
    \includegraphics[scale=0.44]{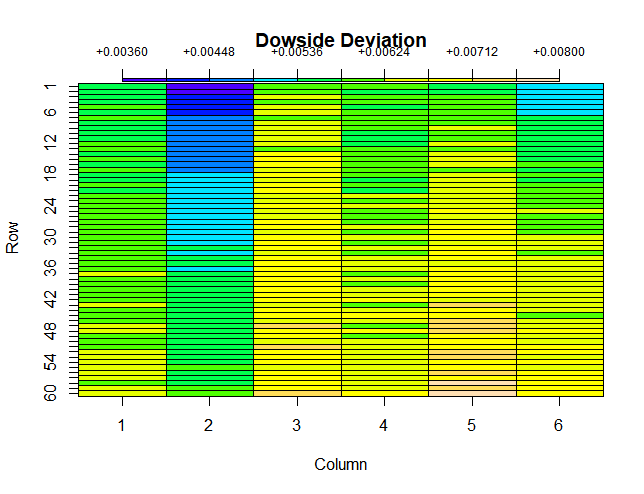}
    \label{cavht3}}
\quad
   \subfigure[VaR of difference series]{
    \includegraphics[scale=0.44]{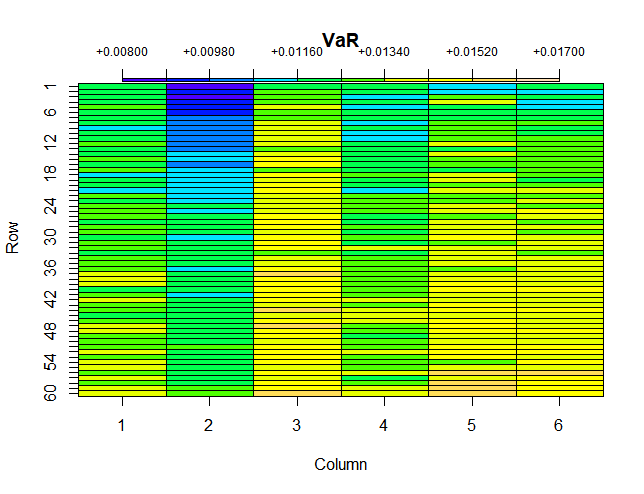}
    \label{cavht4}}
\quad
   \subfigure[CVaR of difference series]{
    \includegraphics[scale=0.44]{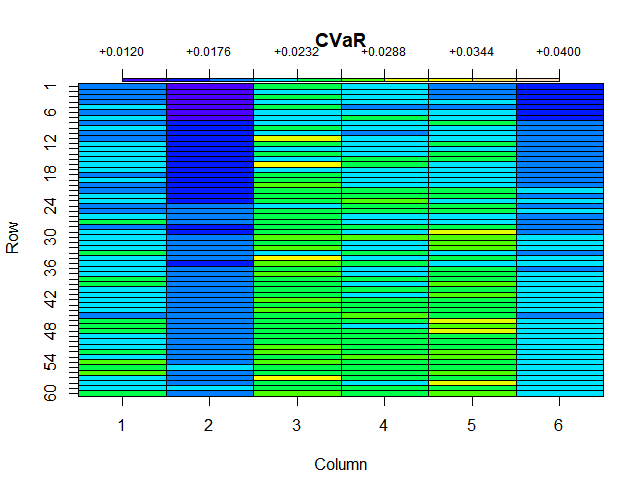}
    \label{cavht5}}
   \end{center}
   \caption{Values of five measures (correlation, standard deviation, downside deviation, VaR, CVaR) for a randomly selected window during the COVID-19 period. The values are plotted for the top 60 stocks (rows) and for the similarity matrices $K_1-K_7$ (columns)}
		\label{analyzing_assets_covid}
\end{figure}

\end{document}